\documentclass[letterpaper,11pt,conference,draftcls,onecolumn]{ieeeconf}
\IEEEoverridecommandlockouts

\usepackage{cite}
\usepackage{amsmath,amssymb,amsfonts}
\usepackage{graphicx}
\usepackage{textcomp}
\usepackage{xcolor}
\usepackage{hyperref}
\usepackage{mathtools}
\usepackage{float}
\usepackage{cite}
\usepackage{xfrac}
\usepackage[euler]{textgreek}

\usepackage{caption}
\usepackage{subcaption}

\usepackage{algpseudocode,algorithm,algorithmicx}

\algrenewcommand\algorithmicrequire{\textbf{Precondition:}}
\algrenewcommand\algorithmicensure{\textbf{Postcondition:}}

\newtheorem{remark}{Remark}
\newtheorem{example}{Example}
\newtheorem{corollary}{Corollary}
\newtheorem{definition}{Definition}
\newtheorem{theorem}{Theorem}
\newtheorem{lemma}[theorem]{Lemma}
\newtheorem{assumption}{Assumption}
\newtheorem{proposition}{Proposition}

\newcommand\norm[1]{\lVert#1\rVert}

\def\BibTeX{{\rm B\kern-.05em{\sc i\kern-.025em b}\kern-.08em
    T\kern-.1667em\lower.7ex\hbox{E}\kern-.125emX}}

\begin{document}
\addtolength{\abovedisplayskip}{-.05cm}
\addtolength{\belowdisplayskip}{-.05cm}
\addtolength{\textfloatsep}{-.5cm}

\title{\LARGE \bf ROBUST SYNCHRONIZATION AND POLICY ADAPTATION FOR NETWORKED HETEROGENEOUS AGENTS}

\author{Miguel F. Arevalo-Castiblanco$^*$, Eduardo Mojica-Nava and, C\'esar A. Uribe
\thanks{ 
$^*$MFAC and emn are With Universidad Nacional de Colombia, Bogota, Colombia, 111321. (email: \{miarevaloc,  eamojican\}@rice.edu). MFAC and CAU are with Rice University, Houston, TX, 77005, USA. (e-mail: \{mfarevalo,  cauribe\}@rice.edu). MFAC is with the Departamento de Ingeniería Eléctrica y Electrónica, Universidad Nacional de Colombia, Bogotá, Colombia 111321, $\{$\texttt{miarevaloc\}@unal.edu.co.}. This research was funded by Minciencias grant number CT 542-2020, “Programa de Investigación en Tecnologías Emergentes para Microrredes Eléctricas Inteligentes con Alta Penetración de Energías Renovables.” The work of MFAC is supported by Fullbright-Estudiante Doctoral Colombiano 2022 and the IEEE Control Systems Society Graduate Collaboration Fellowship. The work of CAU is supported by the National Science Foundation under Grants \#2211815 and No. \#2213568, and the Google Research Award. }}
\maketitle

\begin{abstract}
We propose a robust adaptive online synchronization method for leader-follower networks of nonlinear heterogeneous agents with system uncertainties and input magnitude saturation. Synchronization is achieved using a Distributed input Magnitude Saturation Adaptive Control with Reinforcement Learning (DMSAC-RL), which improves the empirical performance of policies trained on off-the-shelf models using Reinforcement Learning (RL) strategies. The leader observes the performance of a reference model, and followers observe the states and actions of the agents they are connected to, but not the reference model. The leader and followers may differ from the reference model in which the RL control policy was trained. DMSAC-RL uses an internal loop that adjusts the learned policy for the agents in the form of augmented input to solve the distributed control problem, including input-matched uncertainty parameters. We show that the synchronization error of the heterogeneous network is Uniformly Ultimately Bounded (UUB). Numerical analysis of a network of Multiple Input Multiple Output (MIMO) systems supports our theoretical findings. \label{sec:abstract}
\end{abstract}


\section{Introduction} \label{sec1}
The increasing theoretical insights and the efficient implementation of reinforcement learning (RL) methodologies have positioned this framework as a viable option for developing robust and efficient data-driven controllers~\cite{Recht:19, arevalo2024application}. However, gaining a comprehensive theoretical understanding of reinforcement learning remains challenging due to the numerous factors that must be considered when applying it to autonomous agents in real-world scenarios. Among the most significant challenges is the substantial variability in application parameters, which necessitates multiple trials even for the same problem~\cite{Recht:19}. A prime example is addressing the discrepancy between the behaviors exhibited by agents in simulation and those observed in real-world applications, a phenomenon known as the \textit{reality gap}~\cite{Tan:18}.

The \textit{reality gap} arises from the significant difference in cost—both in terms of time and energy—between testing directly on the actual application model and performing simulations~\cite{Koos:13}. Given the iterative nature of learning methods, they are often developed in simulated environments rather than real-world settings. This approach allows for faster simulations at a substantially lower cost. However, \textit{simulations are not reality}. Discrepancies may arise due to errors in system characterization, unmodeled dynamics, or inherent inaccuracies in the model. Consequently, systems that exhibit high performance in simulation may become entirely impractical in real-world applications or may require costly fine-tuning to achieve similar performance in practice~\cite{Nguyen2020}.

The effects of the \textit{reality gap} can be particularly pronounced in Multi-Agent Systems (MAS). In MAS, the interaction among agents forms the cornerstone of cooperative control, which is increasingly recognized as a crucial approach for addressing both current and future critical applications, such as autonomous multi-vehicle systems, resource allocation in networks, synchronization in power systems, and more~\cite{lewis2013cooperative,wang2017cooperative}. The challenge becomes more significant as multiple interacting agents collectively contribute to potential deviations from simulated behaviors~\cite{bucsoniu2010multi,zhang2021multiagent}.

Consensus-based control strategies have become central in cooperative control within MAS, with a wealth of literature supporting this field, beginning with seminal works such as~\cite{Tsitsiklis:89,olfati:07}, and extending to more recent comprehensive reviews~\cite{Cao2012,cortes:19}. Historically, most successful applications of cooperative control have relied on model-based approaches. However, over the past decade, significant efforts have been directed toward addressing the uncertainties inherent in real-world application models. Notably, machine learning-driven approaches have gained prominence in tackling a wide range of challenges in the control of MAS~\cite{Nguyen2020}. These theoretical advancements have found application across diverse domains, from industrial systems, as exemplified by Han et al., to more complex and robust concepts such as the Internet of Battle Things~\cite{Han2021}, where agents communicate and collaborate even in military and adversarial environments~\cite{kott2016internet}.


In Guha et al.\cite{anaswamy:21,guha2021online}, the authors propose a framework that enhances RL-trained policies using adaptive control to address modeling errors and system perturbations. This approach introduces an adaptive control mechanism within the inner loop. At the same time, pre-trained (off-the-shelf) RL policies, specifically those based on the proximal policy optimization algorithm, are applied in the outer loop\cite{schulman2017proximal}. RL techniques can be effectively employed in control problems involving reference models, functioning as reference-based adaptive controllers. The primary role of these controllers is the online adjustment of parameters through adaptive laws to synchronize the system's dynamics with a reference model. In the context of MAS, the concept of Distributed Model Reference Adaptive Control has been integrated with RL techniques (DMRAC-RL) to mitigate the \textit{reality gap} in systems with heterogeneous agents~\cite{Miguel:20}. However, the application of DMRAC-RL methodologies is constrained by their limited consideration of certain inherent aspects of the agents, such as non-linear dynamics, uncertainties, and the specific characteristics of their actuators.

There is a growing interest in integrating RL techniques and adaptive control as a robust framework to deal with these complex environments~\cite{gaudio2019c}. \textit{This paper proposes a framework for robust adaptive synchronization of networked nonlinear heterogeneous agents that use a pre-trained RL policy and an adaptive controller to mitigate model and parameter uncertainties.} Leader-follower synchronization is achieved using a Distributed Input Magnitude Saturation Adaptive Control (DMSAC) that improves the performance of a policy defined by an RL-trained algorithm. Given the difference between a real system and a reference model, this policy is initially adjusted and integrates a distributed reference-based framework for online policy synchronization. The proposed DMSAC-RL uses an internal loop that directly adjusts the policy for agents and complements an external loop in an augmented input to solve the distributed control problem. The control actions resulting from this process include an input saturation component for its correct application in actuators. Moreover, we use optimal modifications \cite{Nguyen2018L} for disturbance suppression of the input-matched uncertainties.

The synchronization of leader agents to the reference model without uncertainties has been previously studied~\cite{anaswamy:21,arevalo2021model}. In contrast to previous control approaches that primarily focus on incorporating distributed control laws for linear systems based on adaptive laws~\cite{Baldi2019}, our framework accounts for nonlinearities and robust parameter handling in the presence of input-matched uncertainties. Similarly, while the work of Guha et al. integrates learning strategies with adaptive laws to enhance the response in nonlinear systems~\cite{anaswamy:21,guha2021online}, it does not address the complexities associated with distributed systems or the uncertainties inherent in MIMO (Multiple-Input Multiple-Output) systems. 

The main contributions of this paper can be summarized as follows:
\begin{itemize}
    \item We define an adaptive synchronization strategy based on a reference model with reinforcement learning for multi-agent control in linear and nonlinear systems.
    \item We propose a robust adaptive distributed law for synchronizing heterogeneous MIMO agents with uncertainties and input magnitude saturation.
    \item We show that the proposed method is uniformly ultimately bounded (UUB) using Lyapunov theory. 
    \item We present numerical evidence of the effectiveness of the proposed approach with simulation results for synchronizing a network of MIMO systems for tracking, unlike works such as those presented in Tao. G without the use of inverse matrices of the adaptive laws for stability analyses~\cite{tao2014multivariable}.
\end{itemize}

The rest of the paper is organized as follows. Section \ref{problem} introduces the optimal leader-follower synchronization problem with a reference model. Section \ref{DMRACRL} shows the proposed MIMO robust distributed MRAC-RL and its stability analysis. Input magnitude saturation analysis is presented in Section~\ref{MSAC}. Section \ref{sims} presents some simulation results to illustrate the performance of the proposed framework. Finally, in Section \ref{conclusion}, some conclusions are drawn.   

\textbf{Notation.} The set of integer numbers is denoted by $\mathbb{Z}$, and the set of real numbers is denoted as $\mathbb{R}$. A matrix and vector are denoted $X$ and $x$, respectively. To denote the reference model, we use $x_m$. We define $x^\top$ and $X^\top$  for the transpose of a vector or a matrix. When the Euclidean norm is needed, we write $\norm{X}^2=\sum_{i=1}^n|{x_i}|^2$. A positive definite matrix is denoted as $X\succ 0$. The trace of a matrix is $\text{tr}(X)$, where $X$ is a square matrix. The eigenvalues of a matrix are denoted with $\lambda$. The estimated values of a parameter $x$ are denoted by $\Tilde{x}$ and its ideal value as $x^*$. The state $y$ for an agent $i$ is denoted as $x_{i,y}$.

\section{Problem Formulation} \label{problem}

We consider a network of $N$ agents, where the dynamics of each agent \mbox{$i \in [1,\cdots, N]$} are modeled as the following dynamical system:
\begin{equation}
    \Dot{x}_i=A_i\sigma_i(x_i)+ B_i\Lambda (u_i+w_i(x_i)), \hspace{0.5cm} i \in \left[1,...,N\right],
    \label{nonlinear-system-imu}
\end{equation}
with $x_i \in \mathbb{R}^n$ is the state of the agent, $\sigma_i:\mathbb{R}^n\xrightarrow{}\mathbb{R}^n$ is a canonical nonlinear map of the states as
\begin{equation}  \sigma_i(x_i)=\left[\psi(x_{i,1}),x_{i,2},x_{i,3},\ldots,x_{i,n}\right]^\top,
    \label{nl_sys}
\end{equation}
with $\psi(x_{i,1})$ acting as a nonlinear \textit{known} function, $u_i$ $\in$ $\mathbb{R}^p$ is the control input, $A_i$ is an \textit{unknown} matrix associated to the agent states, $B_i$ is a \textit{known} input matrix, $w_i\colon\,\mathbb{R}^n\to\mathbb{R}^p$ is a bounded input uncertainty, and $\Lambda$ is an \textit{unknown} efectiveness matrix. 

Agents interact over a network $\mathcal{G}=(V,E)$, where $V= [1,\cdots,N]$ is the set of nodes or agents, and $E$ is the set of edges, such that $(j,i)\in E$ if agent $j$ is an in-neighbor of agent $i$. The adjacency matrix of the graph $\mathcal{G}$ is defined as $\mathcal{A}=[a_{ij}]$ where $a_{ii}=0$ and $a_{ij}=1$ if and only if $(j,i) \in E$, where $i \neq j$. The properties of the graph are specified in the following assumption. \\

\begin{assumption}\label{assum:graphs}
The graph $\mathcal{G}$ is unweighted, directed, and acyclic.
\end{assumption}

The system's heterogeneity is modeled by allowing $A_i{\not=}A_j$ and $B_i{\not=}B_j$. However, we assume that the system dynamics of the agents are sufficiently close, as described in the following assumption.\\

\begin{assumption}[From Proposition 1 in Baldi et al.~\cite{Baldi2019}]\label{assum:coupling-mc}
For every pair of connected agents $i,j\in\left[1,...,N\right]$ with $i\neq j$, there exist matrices $K_{ij}^*\in\mathbb{R}^{n\times p}$ and $K_{rij}^*\in\mathbb{R}^p$, defined as coupling matching conditions, such that
\begin{equation}
    A_j=A_i+{B_i}\Lambda K_{ij}^{*}\; \text{, and }
    B_j={B_i}\Lambda K^{*}_{rij}.
    \label{MC}
\end{equation}
\end{assumption}

Assumption~\ref{assum:coupling-mc} implies that any agent $j$ can match the model of an agent $i$ through appropriate gains. These conditions have been previously used for tracking multi-agent systems in mechanical networks~\cite{arevalo2021model}. Similarly, when we consider the perturbation parameters, we could define the following matching condition. \\

\begin{assumption}\label{assum:uncertainties}
For every pair of connected agents $i,j\in\left[1,...,N\right]$ with $i\neq j$, there exist a matrices $\Theta^*_j\in\mathbb{R}^{n\times p}$, defined as uncertainty matching condition, such that
\begin{equation}
    B_j\Lambda={B_i}\Lambda \Theta^{*}_{j}.
    \label{u_MC}
\end{equation}
\end{assumption}

Moreover, we assume that there is a known reference model that can be understood as an ideal system that describes the unknown dynamics of the agents and for which we have an oracle that can provide off-the-shelf controllers. 
The reference model has the following form
\begin{equation}
    \dot{x}_m=A_m\sigma_m(x_m)+B_mu_m,
    \label{nl_ref}
\end{equation}
where $\sigma_m$ $\in$ $\mathbb{R}^n$ is the reference nonlinear map of $x_m$,  $A_m$ and $B_m$ are its states and input matrices, respectively, and $u_m$ is the control action. The matrix $A_m$ is assumed Hurwitz to have a bounded state trajectory $x_m$ for the reference input signal $u_m$. For notational simplicity, we use $\sigma_i$ or $\sigma_m$ to refer to $\sigma_i(x_i)$ or $\sigma_m(x_m)$ respectively,

Similarly to Assumption~\ref{assum:coupling-mc}, we will assume that while the set of heterogeneous agents has different dynamics from the reference model, such difference is bounded and can be described by a set of matching conditions defined in the next assumption. \\
\begin{assumption}\label{assum:feedback-mc}
For all $i \in \left[1,...,N\right]$ there exists matrices $K_{mi}^*\in\mathbb{R}^{n\times p}$ and $K_{ri}^*\in\mathbb{R}^p$, defined as feedback matching conditions, such that
\begin{equation}
    A_i+ B_i\Lambda K^{*}_{mi}=A_m \; \text{, and } B_i\Lambda K^{*}_{ri}=B_m.
    \label{fmc}
\end{equation}
\end{assumption}

Assumption~\ref{assum:feedback-mc} is required for the existence of a closed-loop system for agents that have access to the reference model.  These conditions have been previously used for adaptive control in aircraft models~\cite{guo2011multivariable}. \\

Additionally, we assume there exists a cost functional $c:X\times U\times \mathbb{N}\xRightarrow \mathbb{R}$ for the definition of the optimal control problem with respect to the reference model as:
\begin{align} \label{opt_prob}
    \min_{u\in\mathcal{U}, \forall t\in\left[0,T\right]}&\int_0^Tc\left(x_m,u_m,t\right)dt, \\
    \text{s.t }\quad & \dot{x}_m={A_m}\sigma_m(x_m)+{B_m}u_m, \hspace{0.5cm} \forall t \in [0,T], \nonumber\\
    & \sigma_m(0)=\sigma_{m0}. \nonumber
\end{align}
In this case, a reinforcement learning strategy is used to generate a control policy $\pi$ such that $u_m(t) =\pi\left(x\right)$ produces the solution of~\eqref{opt_prob}. Note that most RL approaches will formulate Problem~\eqref{opt_prob} with the dynamic of the reference model as a Markov Decision Process (MDP)~\cite{sutton1998introduction}. \\

\begin{remark}
Our goal is not to study the efficiency of RL controllers or to compare RL training methods. Instead, we seek to use a policy trained on a reference model on a system with heterogeneous parameters.
\end{remark}

We define agents as a \text{leader} or leaders as the set of agents with access to the policy $\pi(x_m)$, and the state and control action of the reference model, i.e., $(u_m,x_m)$. Without loss of generality, we assume only one leader exists and denote it as Agent $1$. \\

\begin{assumption}\label{asuum:one_agent}
The graph $\mathcal{G}$  has a spanning tree, where agent $1$ is the root node.
\end{assumption}

A follower agent is defined as not having access to the policy $\pi(\cdot)$, nor the states or actions of the reference model. Follower agents can only observe the states and control actions of their in-neighbors on the network. Figure~\ref{fig:1} shows a network with a leader, a reference model, and four follower agents. Each agent has a controller that takes information from the graph communication, and the control action is regulated by a saturator.

\begin{figure}[t]
\centering
\includegraphics[width = 0.5\textwidth]{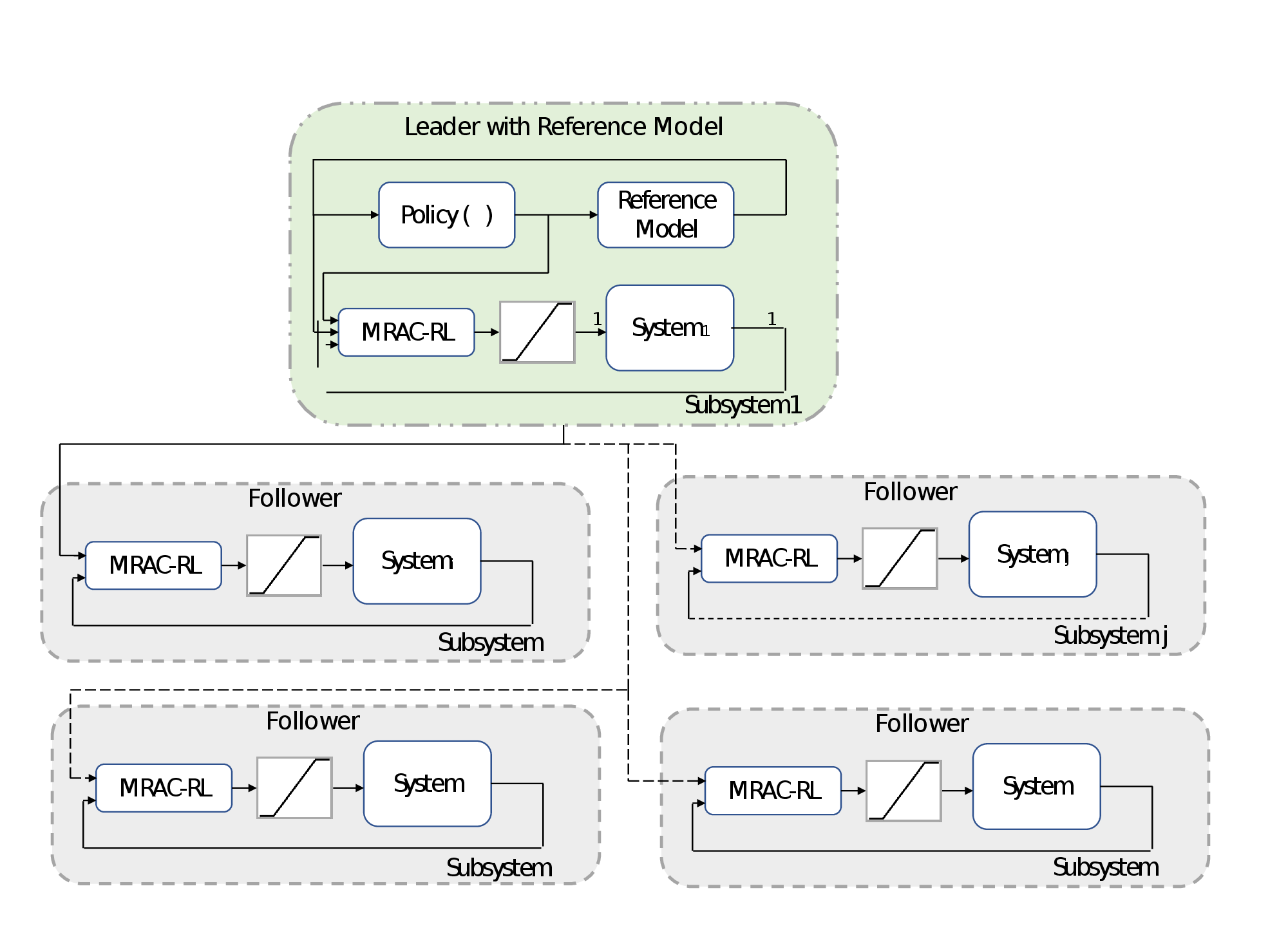}
\caption{Block diagram DMRAC-RL with one leader and four followers. The model trained with the learning strategy and each system, together with its controller with saturation, are represented.}	
\label{fig:1}
\end{figure}

Note that the policy $\pi$ is obtained for the reference model. Therefore, its performance cannot be guaranteed when executed over the leader or follower agents due to the heterogeneity of their models. Moreover, the follower agents are oblivious to the learned RL policy. Thus, our task is to develop local controllers and guarantee that all gents in the network to synchronize their states with the reference trajectory. Formally, we seek to guarantee uniformly ultimately bounded  (UUB) synchronization errors between all agents, as defined below. \\

\begin{definition}\textit{(Uniformly Ultimately Boundedness)}\label{uub}
The solution of a non-autonomous system is said to be uniformly ultimately bounded if, for any $R > 0$, there exists some $r > 0$ independent of $R$ and of the initial time $t_0$ such that 
\begin{equation}
    \norm{x_0}<r\xRightarrow[]{}\norm{x}\leq R, \forall t\geq t_0+T,
\end{equation}
with $T = T(r)$ as a time interval after the initial time $t_0$.
\end{definition}

\begin{example}
We show how discrepancies between the reference model in which the RL policy was trained and the actual model being controlled affect the control system's performance. Figure~\ref{fig:system1a} shows the performance of a control system on an inverted pendulum where the policy was trained with a specified reference model. Additionally, we show the response when the system's parameters differ from the reference model in a certain absolute percentage. The pre-trained policy stabilizes the pendulum around the equilibrium point for the reference model. However, when the linear system parameters differ from those used in the training phase, the system might not converge to equilibrium. In this case, the pre-trained policy does not stabilize the system with a variation above $10 \%$. For a detailed exposition of this phenomenon, see Guha et al.~\cite{anaswamy:21, guha2021online}. 

\begin{figure}[ht]
    \centering
    \includegraphics[width=0.5\textwidth]{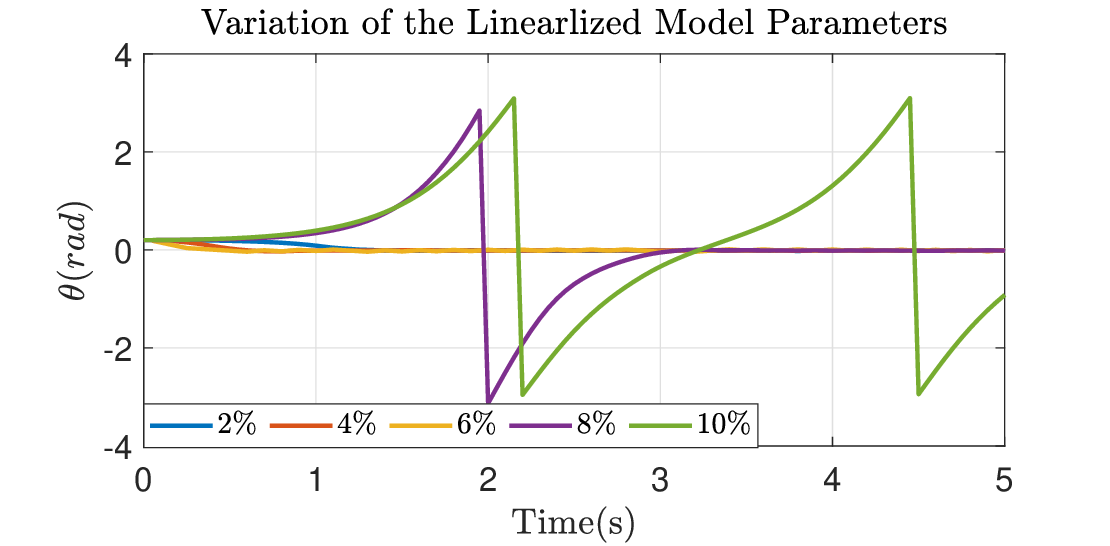}
    \caption{Response of a reinforcement learning algorithm to systems with variation from the parameters used for training in a Nonlinear pendulum model. Each of the lines represents the percentage variation in the system parameters.}
   \label{fig:system1a}
\end{figure}

\end{example}

The following section describes the analysis of the synchronization problem for leader agents based on the described problem formulation.

\section{DMRAC-RL for MIMO Leader Agents}\label{sec3}
With the problem formulation of Section~\ref{problem}, we define the control law for the leading agents considering for the leader the uncertainties $w(x_1)\neq 0$. 

The control law defined for the synchronization of the leader agent is defined as
\begin{equation}
    u_1=K_{m1}\sigma_1+K_{r1}\xi-\Theta_1\phi_1(x_1),
    \label{u1_nl_imu2}
\end{equation} 
where the adaptive gain $K_{m1}$ is the constant associated with the reference states, and $K_{r1}$ is associated with the augmented reference signal, defined as
\begin{equation}
    \xi_1:=u_m-b^mZ^{m\top}_r\left(\sigma_1-\sigma_m\right)+b^m\Upsilon_r^{m\top}e_1,
    \label{xi-def}
\end{equation}
where $Z^m\in\mathbb{R}^{n\times p}$ is a positive definite matrix with the last row of the components of $B_{m}$,
$b^{m}$ is the average of the last row of matrix $B_m$.
$\Upsilon^m\in\mathbb{R}^{n\times p}$ is a matrix corresponding to the last row of matrix $A_m$. The adaptive law $\Theta_1\in\mathbb{R}^{l\times p}$ is used for the suppression of input uncertainty parameters, and $\phi_1\colon \mathbb{R}^n \to \mathbb{R}^p$ is a known bounded basis function. We assume that there exist a $\Theta_1^*$ such that $w_1(x_1)=\Theta_1^{*\top}\phi_1$. Moreover, for an arbitrary $\Theta_1$, we define an approximation error as
\begin{equation}
    \epsilon_1(x_1)=\Theta^{\top}_1\phi(x_1)-w_1(x_1).
\end{equation}

We propose the following dynamical laws for the adaptive parameters
\begin{subequations} \label{al_nl2}
    \begin{align}
        {\dot{K}_{m1}}&=-{\Gamma_m}\sigma_1e^\top_1 P_1{B_1},\\
        {\dot{K}_{r1}}&=-{\Gamma_r}\xi e^\top_1 P_1 {B_1}, \\
        \dot{\Theta}_1&=-{\Gamma_\theta}\phi_1(x_1) e^\top_1 P_1 {B_1}.
    \end{align}    
\end{subequations}
where $\Gamma_m=\Gamma^\top_m\succ 0$, $\Gamma_r=\Gamma^\top\succ 0$, $\Gamma_\theta=\Gamma^\top_\theta\succ 0$ are adaptive gains, and $P_1=P^\top_1\succ 0$ that is the solution of the following linear Lyapunov function
\begin{equation}
    P_1A_{H}+{A^\top_{H}}P_1=-Q, \hspace{0.5cm}Q\succ 0,
    \label{eq6}
\end{equation}
where $A_{H}=M+H\Upsilon_r^{m\top}$, with $B_mb^{m}=H$, and
\begin{equation}
M:=\left[\begin{tabular}{r| l}
$\mathbf{0}_{(n-1)\times 1}$ & $\mathbf{I}_{(n-1)\times (n-1)}$ \\
\hline
\multicolumn{2}{c}{$\mathbf{0}_{1\times n}$} 
\end{tabular}\right].
\label{eq_m}
\end{equation}

Next, we show that dynamic gains in~\eqref{al_nl2} guarantee UUB synchronization error between the leader agent and the reference model. Note that Proposition~\ref{prop:central} extends existing results from SISO to MIMO systems~\cite{arevalo2021model}. 

\begin{proposition}\label{prop:central}
Let Assumptions \ref{assum:feedback-mc} and \ref{asuum:one_agent} hold, and consider the leader agent $1$ with dynamics as in~\eqref{nonlinear-system-imu}, a reference model with dynamics~\eqref{nl_ref}, and the MRAC-RL control law~\eqref{u1_nl_imu2} with adaptive gain laws~\eqref{al_nl2}. Then, the synchronization error between the leader agent and the reference model, i.e., $e_1=x_1-x_m$, is UUB for all initial conditions.
\end{proposition}

\begin{proof}

The error dynamic $e_1=x_1-x_m$ expanded is
\begin{equation}
    \dot{e}_{1}=A_1\sigma_1+ B_1\Lambda(u_1+w_1(x_1))-A_m\sigma_m-B_mu_m.
    \label{ed1}
\end{equation}

Applying control law~\eqref{u1_nl_imu2}
\begin{align*}
    \dot{e}_{1}=A_1\sigma_1+ B_1\Lambda\left(K_{m1}\sigma_1+K_{r1}\xi_1-\Theta_1\phi_1(x_1)+w_1(x_1)\right)-A_m\sigma_m-B_mu_m,
\end{align*}
From~\eqref{xi-def} in $u_m$, we can let the equation in terms of the augmented input $\xi_1$
\begin{align*}
    \dot{e}_{1}&=A_1\sigma_1+ B_1\Lambda\left(K_{m1}\sigma_1+K_{r1}\xi_1-\Theta_1\phi_1(x_1)+w_1(x_1)\right)-A_m\sigma_m\\
    &-B_m\left(\xi_1+b^{m}Z^{m\top}_r\left(\sigma_1-\sigma_m\right)-b^{m}\Upsilon_r^{m\top}e_1\right).
\end{align*}

Adding $\pm A_m\sigma_1$, and grouping similar terms
\begin{align*}
    \dot{e}_{1}&=(A_1-A_m)\sigma_1+ B_1\Lambda\left(K_{m1}\sigma_1+K_{r1}\xi_1-\Theta_1\phi_1(x_1)+w_1(x_1)\right)+A_m(\sigma_1-\sigma_m)\\
    &-B_m\left(\xi_1+b^{m}Z^{m\top}_r\left(\sigma_1-\sigma_m\right)-b^{m}\Upsilon_r^{m\top}e_1\right).
\end{align*}
expanding terms with the $H$ definition,
\begin{align*}
    \dot{e}_{1}&=(A_1-A_m)\sigma_1+ B_1\Lambda\left(K_{m1}\sigma_1+K_{r1}\xi_1-\Theta_1\phi_1(x_1)+w_1(x_1)\right)+A_m(\sigma_1-\sigma_m)-B_m\xi_1\\
    &-HZ^{m\top}_r\left(\sigma_1-\sigma_m\right)+H\Upsilon_r^{m\top}e_1.
\end{align*}

Considering that \mbox{$A_m(\sigma_1-\sigma_m)=Me_1+HZ_r^{m\top}(\sigma_1-\sigma_m)$} from the definition of~\eqref{eq_m}, then
\begin{align*}
    \dot{e}_{1}=Me_1+ B_1\Lambda\left(K_{m1}\sigma_1+K_{r1}\xi_1-\Theta_1\phi_1(x_1)+w_1(x_1)\right)+(A_1-A_m)\sigma_1-B_m\xi_1+H\Upsilon_r^{m\top}e_1.
\end{align*}
with the definition of $A_H$, we have
\begin{align*}
    \dot{e}_{1}=A_He_1+ B_1\Lambda\left(K_{m1}\sigma_1+K_{r1}\xi_1-\Theta_1\phi_1(x_1)+w_1(x_1)\right)+(A_1-A_m)\sigma_1-B_m\xi_1.
\end{align*}

Considering the input uncertainty approximation term $w_1(x)=\Theta^{*}_1\phi_1(x_1)$,
\begin{align*}
    \dot{e}_{1}=A_He_1+ B_1\Lambda\left(K_{m1}\sigma_1+K_{r1}\xi_1-\Theta_1\phi_1(x_1)+\Theta^{*}_1\phi_1(x_1)\right)+(A_1-A_m)\sigma_1-B_m\xi_1.
\end{align*}
with the matching condition~\eqref{fmc}, we obtain
\begin{align*}
    \dot{e}_{1}=A_He_1+ B_1\Lambda\left(K_{m1}\sigma_1+K_{r1}\xi_1-\Theta_1\phi_1(x_1)+\Theta^{*}_1\phi_1(x_1)\right)-B_1\Lambda K^{*}_{m1}\sigma_1-B_1\Lambda K^{*}_{r1}\xi_1.
\end{align*}

Grouping similar terms related with the parameters of the controller, we have
\begin{align}
    \dot{e}_{1} =A_{H}e_{1}+B_1\Lambda \left[(K_{m1}-K^{*}_{m1})\sigma_1+(K_{r1}-K^{*}_{r1})\xi_1-(\Theta_1-\Theta^{*}_1)\phi_1(x_1)\right].
\end{align}

Considering the estimation errors $\tilde{K}_{m1}=K_{m1}-K^*_{m1}$,  $\tilde{K}_{r1}=K_{r1}-K^*_{r1}$, $\tilde{\Theta}_1=\Theta_1-\Theta^*_1$, then the error dynamics is
\begin{align}
    \dot{e}_{1}=A_{H}e_{1}+B_1\Lambda \left[\tilde{K}_{m1}\sigma_1+\tilde{K}_{r1}\xi_1-\tilde{\Theta}_1\phi_1(x_1)\right],
    \label{err_2}
\end{align}

Now, consider the following Lyapunov function
\begin{align}
    V &=e^\top_1P_1e_1+\text{tr}\left (\Lambda\tilde{K}_{m1}\Gamma^{-1}_{m}\tilde{K}^\top_{m1} \right )+\text{tr}\left (\Lambda\tilde{K}_{r1}\Gamma^{-1}_{r}\tilde{K}^\top_{r1} \right )+\text{tr}\left(\Lambda\tilde{\Theta}_1\Gamma ^{-1}_\theta\tilde{\Theta}^\top_1\right).
    \label{lyapunov_1}
\end{align}

The time derivative of~\eqref{lyapunov_1} along the error $e_1$ is
\begin{align}
    \dot{V}=\dot{e}^\top_1 P_1 e_1+e_{1}^\top P_1\dot{e}_{1}+2\text{tr}\left(\Lambda \tilde{K}_{m1}\Gamma^{-1}_{m}\dot{\tilde{K}}^\top_{m1}\right)+2\text{tr}\left(\Lambda \tilde{K}_{r1}\Gamma^{-1}_{r}\dot{\tilde{K}}^\top_{r1}\right)+2\text{tr}\left(\Lambda \tilde{\Theta}_1\Gamma ^{-1}_\theta\dot{\tilde{\Theta}}^\top_1\right),
    \label{dlyapg2}
\end{align}
which expanded through the definition of the error dynamics~\eqref{err_2} gives us
\begin{align*}
    \dot{V}&=\left(A_{H}e_{1}+B_1\Lambda\left[\tilde{K}_{m1}\sigma_1+\tilde{K}_{r1}\xi_{1}-\tilde{\Theta}_1\phi_1(x_1)\right]\right)^\top P_1e_{1}\\
    &+e^\top_{1}P_1\left(A_{H}e_{1}+B_1\Lambda\left[\tilde{K}_{m1}\sigma_1+\tilde{K}_{r1}\xi_{1}-\tilde{\Theta}_1\phi_1(x_1)\right]\right)\\
    &+2\text{tr}\left(\Lambda\tilde{K}_{m1}\Gamma^{-1}_{m}\dot{\tilde{K}}^\top_{m1}\right)+2\text{tr}\left(\Lambda\tilde{K}_{r1}\Gamma^{-1}_{r}\dot{\tilde{K}}^\top_{r1}\right)+2\text{tr}\left(\Lambda\tilde{\Theta}_1\Gamma ^{-1}_\theta\dot{\tilde{\Theta}}^\top_1\right).
\end{align*}

Grouping relative terms associated with the adaptive laws $K_{m1},K_{r1},\Theta_1$, that implies
\begin{align*}
    \dot{V}&=e^\top_{1}A^\top_{H}P_1e_{1}+e^\top_{1}P_1A_{H}e_{1}+2\left[e^\top_{1}P_1B_1\Lambda\tilde{K}_{m1}\sigma_1+\text{tr}\left(\Lambda\tilde{K}_{m1}\Gamma^{-1}_{m}\dot{\tilde{K}}^\top_{m1}\right)\right]\\
    &+2\left[e^\top_{1}P_1B_1\Lambda\tilde{K}_{r1}\xi_{1}+\text{tr}\left(\Lambda\tilde{K}_{r1}\Gamma^{-1}_{r}\dot{\tilde{K}}^\top_{r1}\right)\right]\\
    &+2\left[e^\top_{1}P_1B_1\Lambda\tilde{\Theta}_{1}\phi_1(x_1)+\text{tr}\left(\Lambda\tilde{\Theta}_1\Gamma ^{-1}_\theta\dot{\tilde{\Theta}}^\top_1\right)\right],
\end{align*}
considering the trace property of tr$(CD^\top)=D^\top C$, with $C,D\in\mathbb{R}^n$, and the definition of $P_1$, we can rewrite the derivative as
\begin{align*}
    \dot{V}&=-e^\top_{1}Qe_{1}+2\text{tr}\left(\Lambda\tilde{K}_{m1}\sigma_1e^\top_{1}P_1B_1+\Lambda\tilde{K}_{m1}\Gamma^{-1}_{m}\dot{\tilde{K}}^\top_{m1}\right)+2\text{tr}\left(\Lambda\tilde{K}_{r1}\xi_{1}e^\top_{1}P_1B_1+\Lambda\tilde{K}_{r1}\Gamma^{-1}_{r}\dot{\tilde{K}}^\top_{r1}\right) \\
    &+2\text{tr}\left(\Lambda\tilde{\Theta}_1\phi_1(x_1)e^\top_{1}P_1B_1+\Lambda\tilde{\Theta}_1\Gamma ^{-1}_\theta\dot{\tilde{\Theta}}^\top_1\right),
\end{align*}
factorizing $\Lambda$, we can obtain,
\begin{align*}
    \dot{V}&=-e^\top_{1}Qe_{1}+2\Lambda\text{tr}\left(\tilde{K}_{m1}\sigma_1e^\top_{1}P_1B_1+\tilde{K}_{m1}\Gamma^{-1}_{m}\dot{\tilde{K}}^\top_{m1}\right)
    +2\Lambda\text{tr}\left(\tilde{K}_{r1}\xi_{1}e^\top_{1}P_1B_1+\tilde{K}_{r1}\Gamma^{-1}_{r}\dot{\tilde{K}}^\top_{r1}\right) \\
    &+2\Lambda\text{tr}\left(\tilde{\Theta}_1\phi_1(x_1)e^\top_{1}P_1B_1+\tilde{\Theta}_1\Gamma ^{-1}_\theta\dot{\tilde{\Theta}}^\top_1\right),
\end{align*}
grouping in terms of the estimators $\tilde{K}_{m1},\tilde{K}_{r1},\tilde{\Theta}_1$, we have
\begin{align*}
    \dot{V}&=-e^\top_{1}Qe_{1}+2\Lambda\text{tr}\left(\tilde{K}_{m1}\left(\sigma_1e^\top_{1}P_1B_1+\Gamma^{-1}_{m}\dot{\tilde{K}}^\top_{m1}\right)\right)
    +2\Lambda\text{tr}\left(\tilde{K}_{r1}\left(\xi_{1}e^\top_{1}P_1B_1+\Gamma^{-1}_{r}\dot{\tilde{K}}^\top_{r1}\right)\right) \\
    &+2\Lambda\text{tr}\left(\tilde{\Theta}_1\left(\phi_1(x_1)e^\top_{1}P_1B_1+\Gamma ^{-1}_\theta\dot{\tilde{\Theta}}^\top_1\right)\right).
\end{align*}

Because $K_{m1},K_{r1},\Theta_1$ are constants, therefore $\dot{\tilde{K}}_{m1}=\dot{K}_{m1}$, $\dot{\tilde{K}}_{r1}=\dot{K}_{r1}$ and $\dot{\tilde{\Theta}}_1=\dot{\Theta}_1$, then we can reduce to
\begin{equation}
    \dot{V}=-e^\top_1Qe_1\leq-\lambda_\text{min}\left(Q\right)\norm{e_1}^2\leq 0,
\end{equation}
where using Barbalat's lemma~\cite{farkas2016variations} and with definition~\ref{uub}, the synchronization error is UUB with \eqref{lyapunov_1} as a valid Lyapunov function.
\end{proof}

Along with the analysis for leaders, the next section presents the procedures for synchronization in follower agents.

\section{Distributed Model Reference Adaptive Control with Reinforcement Learning} \label{DMRACRL}
This section presents the main contribution of this work of a DMRAC-RL for follower MIMO agents with input uncertainty parameters. We consider a network of heterogeneous agents, where each agent is represented by dynamics \eqref{nonlinear-system-imu}. In the distributed case, the control law used for the synchronization of agents that do not have communication with the reference is
\begin{align}
    u_i=\sum_{j=1}^{N}a_{ij}{{K}_{ij}}\sigma_j(x_j)+K_{mi}\Xi_{i}+\sum_{j=1}^{N}a_{ij}K_{rij}\xi_{ij}
    +\sum_{j=1}^{N}a_{ij}\Theta_{j}\phi_{j}-\Theta_i\phi_i(x_i),    \label{u_th3}
\end{align}
with the synchronization error $e_{ij}=x_i-x_j$, the augmented input $\Xi_{i}=\sum_{j=1}^{N}a_{ij}\left(\sigma_i-\sigma_j\right)$, and
\begin{equation}
    \xi_{ij}:=u_j-b^jZ^{j\top}_r\left(\sigma_i-\sigma_j\right)+b^j\Upsilon_r^{j\top}e_{ij},
    \label{augmented}
\end{equation}
with $Z^{j}_r$ as a positive definite matrix, $\Upsilon^j_r$ as a $n-$dimensional matrix picked with strictly negative components, and $b^j$ as the average of the elements of the last row of the matrix $B_j$. The adaptive laws used in this case are
\begin{subequations}
\begin{align} \label{al}
    {\dot{K}_{ij}}=&-\Gamma_{ij}\sigma_j(x_j)e^\top_{ij}P_i {B_i}, \\
    {\dot{K}_{mi}}=&-\Gamma_m\Xi_{i}e^\top_{ij}P_i{B_i},\\
    \dot{K}_{rij}=&-\Gamma_{r}\xi_{ij} e^\top_{ij}P_i {B_i}, \\
    \dot{\Theta}_{j}=&-\Gamma_{\phi}\phi_j(x_j) e^\top_{ij}P_i {B_i}, \\
    \dot{\Theta}_{i}=&-\Gamma_{\theta}\phi_i(x_i) e^\top_{ij}P_i {B_i}.
\end{align}   
\end{subequations}

with $\Gamma_{ij}\succ 0$, $\Gamma_{m}\succ 0$, $\Gamma_{r}\succ 0$, $\Gamma_{\theta}\succ 0$, $\Gamma_{\phi}\succ 0$, and $P_i$ that is the solution of the linear Lyapunov function
\begin{equation}
    P_iA_{Hj}+{A^\top_{Hj}}P_i=-Q_i, \hspace{0.5cm}Q_i\succ 0,
    \label{eq6_2}
\end{equation}
where $\sum_{j=1}^{N}a_{ij}A_{Hj}=M+\sum_{j=1}^{N}a_{ij}H_j\Upsilon_r^{m\top}$, with $B_j\Lambda b^j=H_j$, in the case with just one reference model $P_i=P_1$. The following lemma presents the stability results for a general distributed case. \\

\begin{lemma} \label{lem_nl}
Let Assumptions \ref{assum:graphs}-\ref{asuum:one_agent} hold. Consider a network of systems~\eqref{nonlinear-system-imu} with a reference system~\eqref{nl_ref}, and control and adaptive laws \eqref{u_th3}--\eqref{al}. Then, function 
\begin{align}
    V &=\sum_{i=1}^{N}\sum_{j=1}^{N}a_{ij}e_{ij}^\top P_ie_{ij}+ \sum_{i=1}^{N} \text{tr}\left (\Lambda\tilde{K}_{mi}\Gamma_{m}\tilde{K}^\top_{mi} \right ) + \sum_{i=1}^{N} \sum_{j=1}^{N}a_{ij} \text{tr}\left (\Lambda{\tilde{K}_{ij}}\Gamma_{ij}\tilde{K}^\top_{ij} \right ) 
     \nonumber \\
    &+ \sum_{i=1}^{N}\sum_{j=1}^{N}a_{ij}\text{tr}\left (\Lambda{\tilde{K}_{rij}}\Gamma_{r}\tilde{K}^\top_{rij} \right )+ \sum_{i=1}^{N} \sum_{j=1}^{N}a_{ij}\; \text{tr}\left (\Lambda{\tilde{\Theta}_{j}}\Gamma_{\phi}\tilde{\Theta}^\top_{j} \right )
    +\sum_{i=1}^{N}\text{tr}(\Lambda\tilde{\Theta}_i\Gamma^{-1}_{\Theta}\tilde{\Theta}^\top_i),
    \label{lyapunov_nl}
\end{align}

is a valid Lyapunov function.
\end{lemma} 

\begin{proof}
With the error $e_i=x_i-x_m$ defined in Proposition~\ref{prop:central}. In this case, the error dynamic for an agent~$i$ connected to an agent $j$, expanded is
\begin{equation}
    \dot{e}_{ij}=A_i\sigma_i(x_i)+B_i\Lambda (u_i+w_i(x_i))-A_j\sigma_j-B_j\Lambda (u_j+\phi_j).
\end{equation}

Analyzing the error for an agent $i$ and its neighbors in the network, the synchronization error can be defined as
\begin{align}
    \sum_{j=1}^{N}a_{ij}\dot{e}_{ij}=A_i\sigma_i(x_i)+B_i\Lambda (u_i+w_i(x_i))-\sum_{j=1}^{N}a_{ij}A_j\sigma_j
    -\sum_{j=1}^{N}a_{ij}B_j\Lambda (u_j+\phi_j),
\end{align}
where using the control law~\eqref{u_th3}, we can have
\begin{align}
    \sum_{j=1}^{N}a_{ij}\dot{e}_{ij}&=A_i\sigma_i(x_i)+B_i\Lambda (\sum_{j=1}^{N}a_{ij}{{K}_{ij}}\sigma_j(x_j)+K_{mi}\Xi_{i}+\sum_{j=1}^{N}a_{ij}K_{rij}\xi_{ij}+\sum_{j=1}^{N}a_{ij}\Theta_{j}\phi_{j}-\Theta_i\phi_i(x_i)+w_i(x_i)) \nonumber \\
    &-\sum_{j=1}^{N}a_{ij}A_j\sigma_j-\sum_{j=1}^{N}a_{ij}B_j\Lambda (u_j+\phi_j),
\end{align}
expanding the terms related with $B_j$,
\begin{align}
    \sum_{j=1}^{N}a_{ij}\dot{e}_{ij}&=A_i\sigma_i(x_i)+B_i\Lambda (\sum_{j=1}^{N}a_{ij}{{K}_{ij}}\sigma_j(x_j)+K_{mi}\Xi_{i}+\sum_{j=1}^{N}a_{ij}K_{rij}\xi_{ij}\nonumber \\
    &+\sum_{j=1}^{N}a_{ij}\Theta_{j}\phi_{j}-\Theta_i\phi_i(x_i)+w_i(x_i))-\sum_{j=1}^{N}a_{ij}A_j\sigma_j -\sum_{j=1}^{N}a_{ij}B_j\Lambda u_j - \sum_{j=1}^{N}a_{ij}B_j\Lambda \phi_j.
\end{align}
considering then, the augmented input~\eqref{augmented}, we have
\begin{align}
    \sum_{j=1}^{N}a_{ij}\dot{e}_{ij}&=A_i\sigma_i(x_i)+B_i\Lambda (\sum_{j=1}^{N}a_{ij}{{K}_{ij}}\sigma_j(x_j)+K_{mi}\Xi_{i}+\sum_{j=1}^{N}a_{ij}K_{rij}\xi_{ij}+\sum_{j=1}^{N}a_{ij}\Theta_{j}\phi_{j}-\Theta_i\phi_i(x_i)+w_i(x_i)) \nonumber \\
    &-\sum_{j=1}^{N}a_{ij}A_j\sigma_j-\sum_{j=1}^{N}a_{ij}B_j\Lambda (\xi_{ij}+b^jZ^{j\top}_r\left(\sigma_i-\sigma_j\right)-b^j\Upsilon_r^{j\top}e_{ij})  - \sum_{j=1}^{N}a_{ij}B_j\Lambda \phi_j,
\end{align}
with the definition of $H_j=B_j\Lambda b^j$,
\begin{align}
    \sum_{j=1}^{N}a_{ij}\dot{e}_{ij}&=A_i\sigma_i(x_i) \nonumber \\
    &+B_i\Lambda (\sum_{j=1}^{N}a_{ij}{{K}_{ij}}\sigma_j(x_j)+K_{mi}\Xi_{i}+\sum_{j=1}^{N}a_{ij}K_{rij}\xi_{ij} +\sum_{j=1}^{N}a_{ij}\Theta_{j}\phi_{j}-\Theta_i\phi_i(x_i)+w_i(x_i))\nonumber \\
    &-\sum_{j=1}^{N}a_{ij}A_j\sigma_j-\sum_{j=1}^{N}a_{ij}B_j\Lambda\xi_{ij}-\sum_{j=1}^{N}H_jZ^{j\top}_r\left(\sigma_i-\sigma_j\right)+\sum_{j=1}^{N}H_j\Upsilon_r^{j\top}e_{ij}- \sum_{j=1}^{N}a_{ij}B_j\Lambda \phi_j.
\end{align}

Using the coupling matching conditions~\eqref{MC}, and replacing $A_j$, and $B_j$
\begin{align}
    \sum_{j=1}^{N}a_{ij}\dot{e}_{ij}&=A_i\sigma_i(x_i) \nonumber\\
    &+B_i\Lambda (\sum_{j=1}^{N}a_{ij}{{K}_{ij}}\sigma_j(x_j)+K_{mi}\Xi_{i}+\sum_{j=1}^{N}a_{ij}K_{rij}\xi_{ij}+\sum_{j=1}^{N}a_{ij}\Theta_{j}\phi_{j}-\Theta_i\phi_i(x_i)+w_i(x_i)) \nonumber \\
    &-\sum_{j=1}^{N}a_{ij}(A_i+{B_i}\Lambda K_{ij}^{*})\sigma_j-\sum_{j=1}^{N}a_{ij}{B_i}\Lambda K^{*}_{rij}\xi_{ij}-\sum_{j=1}^{N}H_jZ^{j\top}_r\left(\sigma_i-\sigma_j\right)+\sum_{j=1}^{N}H_j\Upsilon_r^{j\top}e_{ij} \nonumber \\
    &- \sum_{j=1}^{N}a_{ij}B_j\Lambda \phi_j.
\end{align}
Expanding the $(A_i+{B_i}\Lambda K_{ij}^{*})$ term
\begin{align}
    \sum_{j=1}^{N}a_{ij}\dot{e}_{ij}&=A_i\sigma_i(x_i) \nonumber \\
    &+B_i\Lambda (\sum_{j=1}^{N}a_{ij}{{K}_{ij}}\sigma_j(x_j)+K_{mi}\Xi_{i}+\sum_{j=1}^{N}a_{ij}K_{rij}\xi_{ij}+\sum_{j=1}^{N}a_{ij}\Theta_{j}\phi_{j}-\Theta_i\phi_i(x_i)+w_i(x_i)) \nonumber \\
    &-\sum_{j=1}^{N}a_{ij}A_i\sigma_j-\sum_{j=1}^{N}{B_i}\Lambda K_{ij}^{*}\sigma_j-\sum_{j=1}^{N}a_{ij}{B_i}\Lambda K^{*}_{rij}\xi_{ij}-\sum_{j=1}^{N}H_jZ^{j\top}_r\left(\sigma_i-\sigma_j\right)+\sum_{j=1}^{N}H_j\Upsilon_r^{j\top}e_{ij} \nonumber \\
    &- \sum_{j=1}^{N}a_{ij}B_j\Lambda \phi_j,
\end{align}
grouping then with respect to $A_i$,
\begin{align}
    \sum_{j=1}^{N}a_{ij}\dot{e}_{ij}&=A_i\sum_{j=1}^{N}a_{ij}(\sigma_i(x_i)-\sigma_j(x_j))+B_i\Lambda (\sum_{j=1}^{N}a_{ij}{{K}_{ij}}\sigma_j(x_j)+K_{mi}\Xi_{i}+\sum_{j=1}^{N}a_{ij}K_{rij}\xi_{ij} \nonumber \\
    &+\sum_{j=1}^{N}a_{ij}\Theta_{j}\phi_{j}-\Theta_i\phi_i(x_i)+w_i(x_i))-\sum_{j=1}^{N}{B_i}\Lambda K_{ij}^{*}\sigma_j-\sum_{j=1}^{N}a_{ij}{B_i}\Lambda K^{*}_{rij}\xi_{ij} \nonumber \\
    &-\sum_{j=1}^{N}H_jZ^{j\top}_r\left(\sigma_i-\sigma_j\right)+\sum_{j=1}^{N}H_j\Upsilon_r^{j\top}e_{ij}- \sum_{j=1}^{N}a_{ij}B_j\Lambda \phi_j.
\end{align}

Now using the feedback matching conditions~\eqref{fmc} for $A_i$, we can have
\begin{align}
    \sum_{j=1}^{N}a_{ij}\dot{e}_{ij}&=(A_m-B_i\Lambda K^{*}_{mi})\sum_{j=1}^{N}a_{ij}(\sigma_i-\sigma_j)+B_i\Lambda (\sum_{j=1}^{N}a_{ij}{{K}_{ij}}\sigma_j(x_j)+K_{mi}\Xi_{i}+\sum_{j=1}^{N}a_{ij}K_{rij}\xi_{ij} \nonumber \\
    &+\sum_{j=1}^{N}a_{ij}\Theta_{j}\phi_{j}-\Theta_i\phi_i(x_i)+w_i(x_i))-\sum_{j=1}^{N}{B_i}\Lambda K_{ij}^{*}\sigma_j-\sum_{j=1}^{N}a_{ij}{B_i}\Lambda K^{*}_{rij}\xi_{ij} \nonumber \\
    &-\sum_{j=1}^{N}H_jZ^{j\top}_r\left(\sigma_i-\sigma_j\right)+\sum_{j=1}^{N}H_j\Upsilon_r^{j\top}e_{ij}- \sum_{j=1}^{N}a_{ij}B_j\Lambda \phi_j,
\end{align}
then, with the definition of
\begin{equation}
    A_m\sum_{j=1}^{N}a_{ij}(\sigma_i-\sigma_j)=M\sum_{j=1}^{N}a_{ij}e_{ij}+\sum_{j=1}^{N}a_{ij}H_jZ^{j\top}_r\left(\sigma_i-\sigma_j\right),
\end{equation}
and expanding the terms related with the difference of $(\sigma_i-\sigma_j)$
\begin{align}
    \sum_{j=1}^{N}a_{ij}\dot{e}_{ij}&=A_m\sum_{j=1}^{N}a_{ij}(\sigma_i-\sigma_j)-B_i\Lambda K^{*}_{mi}\sum_{j=1}^{N}a_{ij}(\sigma_i-\sigma_j) \nonumber\\
    &+B_i\Lambda \left(\sum_{j=1}^{N}a_{ij}{{K}_{ij}}\sigma_j(x_j)+K_{mi}\Xi_{i}+\sum_{j=1}^{N}a_{ij}K_{rij}\xi_{ij}+\sum_{j=1}^{N}a_{ij}\Theta_{j}\phi_{j}-\Theta_i\phi_i(x_i)+w_i(x_i)\right) \nonumber \\
    &-\sum_{j=1}^{N}{B_i}\Lambda K_{ij}^{*}\sigma_j-\sum_{j=1}^{N}a_{ij}{B_i}\Lambda K^{*}_{rij}\xi_{ij}-\sum_{j=1}^{N}H_jZ^{j\top}_r\left(\sigma_i-\sigma_j\right)+\sum_{j=1}^{N}H_j\Upsilon_r^{j\top}e_{ij}- \sum_{j=1}^{N}a_{ij}B_j\Lambda \phi_j,
\end{align}
then we have,
\begin{align}
    \sum_{j=1}^{N}a_{ij}\dot{e}_{ij}&=M\sum_{j=1}^{N}a_{ij}e_{ij}-B_i\Lambda K^{*}_{mi}\sum_{j=1}^{N}a_{ij}(\sigma_i-\sigma_j)+B_i\Lambda (\sum_{j=1}^{N}a_{ij}{{K}_{ij}}\sigma_j(x_j)+K_{mi}\Xi_{i}+\sum_{j=1}^{N}a_{ij}K_{rij}\xi_{ij} \nonumber \\
    &+\sum_{j=1}^{N}a_{ij}\Theta_{j}\phi_{j}-\Theta_i\phi_i(x_i)+w_i(x_i))-\sum_{j=1}^{N}{B_i}\Lambda K_{ij}^{*}\sigma_j-\sum_{j=1}^{N}a_{ij}{B_i}\Lambda K^{*}_{rij}\xi_{ij}+\sum_{j=1}^{N}H_j\Upsilon_r^{j\top}e_{ij} \nonumber \\
    &- \sum_{j=1}^{N}a_{ij}B_j\Lambda \phi_j,
\end{align}
grouping by $e_{ij}$
\begin{align}
    \sum_{j=1}^{N}a_{ij}\dot{e}_{ij}&=\sum_{j=1}^{N}a_{ij}(M+H_j\Upsilon_r^{j\top})e_{ij}-B_i\Lambda K^{*}_{mi}\sum_{j=1}^{N}a_{ij}(\sigma_i-\sigma_j) \nonumber \\
    &+B_i\Lambda \left(\sum_{j=1}^{N}a_{ij}{{K}_{ij}}\sigma_j(x_j)+K_{mi}\Xi_{i}+\sum_{j=1}^{N}a_{ij}K_{rij}\xi_{ij}+\sum_{j=1}^{N}a_{ij}\Theta_{j}\phi_{j}-\Theta_i\phi_i(x_i)+w_i(x_i)\right) \nonumber \\
    &-\sum_{j=1}^{N}{B_i}\Lambda K_{ij}^{*}\sigma_j-\sum_{j=1}^{N}a_{ij}{B_i}\Lambda K^{*}_{rij}\xi_{ij}- \sum_{j=1}^{N}a_{ij}B_j\Lambda \phi_j,
\end{align}

with the definition of $\sum_{j=1}^{N}a_{ij}A_{Hj}=M+\sum_{j=1}^{N}a_{ij}H_j\Upsilon_r^{j\top}$, we have
\begin{align}
    \sum_{j=1}^{N}a_{ij}\dot{e}_{ij}&=\sum_{j=1}^{N}a_{ij}A_{Hj}e_{ij}-B_i\Lambda K^{*}_{mi}\sum_{j=1}^{N}a_{ij}(\sigma_i-\sigma_j) \nonumber \\
    &+B_i\Lambda \left(\sum_{j=1}^{N}a_{ij}{{K}_{ij}}\sigma_j(x_j)+K_{mi}\Xi_{i}+\sum_{j=1}^{N}a_{ij}K_{rij}\xi_{ij}+\sum_{j=1}^{N}a_{ij}\Theta_{j}\phi_{j}-\Theta_i\phi_i(x_i)+w_i(x_i)\right) \nonumber \\
    &-\sum_{j=1}^{N}{B_i}\Lambda K_{ij}^{*}\sigma_j-\sum_{j=1}^{N}a_{ij}{B_i}\Lambda K^{*}_{rij}\xi_{ij}- \sum_{j=1}^{N}a_{ij}B_j\Lambda \phi_j.
\end{align}

Considering the input uncertainty approximation terms as $w_i(x)=\Theta^{*}_i\phi_i(x_i)$,
\begin{align}
    \sum_{j=1}^{N}a_{ij}\dot{e}_{ij}&=\sum_{j=1}^{N}a_{ij}A_{Hj}e_{ij}-B_i\Lambda K^{*}_{mi}\sum_{j=1}^{N}a_{ij}(\sigma_i-\sigma_j) \nonumber \\
    &+B_i\Lambda \left(\sum_{j=1}^{N}a_{ij}{{K}_{ij}}\sigma_j(x_j)+K_{mi}\Xi_{i}+\sum_{j=1}^{N}a_{ij}K_{rij}\xi_{ij}+\sum_{j=1}^{N}a_{ij}\Theta_{j}\phi_{j}-\Theta_i\phi_i(x_i)+\Theta^{*}_i\phi_i(x_i)\right) \nonumber \\
    &-\sum_{j=1}^{N}{B_i}\Lambda K_{ij}^{*}\sigma_j-\sum_{j=1}^{N}a_{ij}{B_i}\Lambda K^{*}_{rij}\xi_{ij}- \sum_{j=1}^{N}a_{ij}B_j\Lambda\phi_j(x_j).
\end{align}
Now, using the uncertainty matching condition~\eqref{u_MC}, we can have
\begin{align}
    \sum_{j=1}^{N}a_{ij}\dot{e}_{ij}&=\sum_{j=1}^{N}a_{ij}A_{Hj}e_{ij}-B_i\Lambda K^{*}_{mi}\sum_{j=1}^{N}a_{ij}(\sigma_i-\sigma_j) \nonumber \\
    &+B_i\Lambda \left(\sum_{j=1}^{N}a_{ij}{{K}_{ij}}\sigma_j(x_j)+K_{mi}\Xi_{i}+\sum_{j=1}^{N}a_{ij}K_{rij}\xi_{ij}+\sum_{j=1}^{N}a_{ij}\Theta_{j}\phi_{j}-\Theta_i\phi_i(x_i)+\Theta^{*}_i\phi_i(x_i)\right) \nonumber \\
    &-\sum_{j=1}^{N}{B_i}\Lambda K_{ij}^{*}\sigma_j-\sum_{j=1}^{N}a_{ij}{B_i}\Lambda K^{*}_{rij}\xi_{ij}- \sum_{j=1}^{N}a_{ij}B_i\Lambda\Theta^*_j\phi_j(x_j).
\end{align}
grouping according to $B_i$,
\begin{align}
    &\sum_{j=1}^{N}a_{ij}\dot{e}_{ij}=\sum_{j=1}^{N}a_{ij}A_{Hj}e_{ij} \nonumber \\
    &+\sum_{j=1}^{N}B_i\Lambda\left(K_{mi}\left(\sigma_i-\sigma_j\right)
    -K^{*}_{mi}\left(\sigma_i-\sigma_j\right)+{{K}_{ij}}\sigma_j(x_j)-K_{ij}^{*}\sigma_j+K_{rij}\xi_{ij}
    -K^{*}_{rij}\xi_{ij}+\Theta_{j}\phi_{j}\right. \nonumber \\
    &\left.-\Theta^*_j\phi_j(x_j)-\Theta_i\phi_i(x_i)+\Theta^{*}_i\phi_i(x_i)\right).
\end{align}

Likewise, we define the estimation errors $\tilde{K}_{ij}=K_{ij}-K^*_{ij}$, $\tilde{K}_{mi}=K_{mi}-K^*_{mi}$,  $\tilde{K}_{rij}=K_{rij}-K^*_{rij}$, $\tilde{\Theta}_j=\Theta_j-\Theta^*_j$, $\tilde{\Theta}_i=\Theta_i-\Theta^*_i$, the error dynamics can be written as
\begin{align}
    \sum_{j=1}^{N}a_{ij}\dot{e}_{ij}=\sum_{j=1}^{N}a_{ij}A_{Hj}e_{ij}
    +\sum_{j=1}^{N}a_{ij}B_i\Lambda\left(\tilde{K}_{mi}\left(\sigma_i-\sigma_j\right)+\tilde{K}_{ij}\sigma_j(x_j)+\tilde{K}_{rij}\xi_{ij}
    +\tilde{\Theta}_j\phi_{j}(x_j)-\tilde{\Theta}_i\phi_i(x_i)\right).
\end{align}

Now, consider the Lyapunov function~\eqref{lyapunov_nl}. The time derivative is
\begin{align}
    \dot{V}&=\sum_{i=1}^{N}\sum_{j=0}^{N}\dot{e}_{ij}^\top P_ie_{ij}+\sum_{i=1}^{N}\sum_{j=0}^{N}e_{ij}^\top P_i\dot{e}_{ij}
    +2\sum_{i=1}^{N}\sum_{j=1}^{N}a_{ij}\text{tr}\left(\Lambda\tilde{K}^\top_{ij}\Gamma^{-1}_{ij}\dot{\tilde{K}}_{ij}\right) \nonumber \\
    &+2\sum_{i=1}^{N}\text{tr}\left(\Lambda\tilde{K}^\top_{mi}\Gamma^{-1}_{m}\dot{\tilde{K}}_{mi}\right)
    +2\sum_{i=1}^{N}\sum_{j=1}^{N}a_{ij}\text{tr}\left(\Lambda\tilde{K}_{rij}\Gamma^{-1}_{m}\dot{\tilde{K}}_{rij}\right)-2\sum_{i=1}^{N}\text{tr}\left(\Lambda\tilde{\Theta}_i\Gamma^{-1}_{\theta}\dot{\tilde{\Theta}}_{i}\right) \nonumber \\
    &+2\sum_{j=1}^{N}\text{tr}\left(\Lambda\tilde{\Theta}_j\Gamma^{-1}_{\phi}\dot{\tilde{\Theta}}_{j}\right),
    \label{lyapdot}
\end{align}
which expanded through the definition of the error dynamics, is
\begin{align}
    \dot{V}&=\sum_{i=1}^{N}\left(\sum_{j=0}^{N}a_{ij}A_{Hj}e_{ij}+\sum_{j=0}^{N}a_{ij}B_i\Lambda\left(\tilde{K}_{mi}\left(\sigma_i-\sigma_j\right)+\tilde{K}_{ij}\sigma_j(x_j)+\tilde{K}_{rij}\xi_{ij}+\tilde{\Theta}_j\phi_{j}(x_j)-\tilde{\Theta}_i\phi_i(x_i)\right)\right)^\top P_ie_{ij} \nonumber \\
    &+\sum_{i=1}^{N}e_{ij}^\top P_i\left(\sum_{j=0}^{N}a_{ij}A_{Hj}e_{ij}+\sum_{j=0}^{N}a_{ij}B_i\Lambda\left(\tilde{K}_{mi}\left(\sigma_i-\sigma_j\right)+\tilde{K}_{ij}\sigma_j(x_j)+\tilde{K}_{rij}\xi_{ij}
    +\tilde{\Theta}_j\phi_{j}(x_j)-\tilde{\Theta}_i\phi_i(x_i)\right)\right) \nonumber \\
    &+2\sum_{i=1}^{N}\sum_{j=1}^{N}a_{ij}\text{tr}\left(\Lambda\tilde{K}^\top_{ij}\Gamma^{-1}_{ij}\dot{\tilde{K}}_{ij}\right)+2\sum_{i=1}^{N}\text{tr}\left(\Lambda\tilde{K}^\top_{mi}\Gamma^{-1}_{m}\dot{\tilde{K}}_{mi}\right)+2\sum_{i=1}^{N}\sum_{j=1}^{N}\text{tr}\left(\Lambda\tilde{K}_{rij}\Gamma^{-1}_{m}\dot{\tilde{K}}_{rij}\right) \nonumber \\
    &-2\sum_{i=1}^{N}\text{tr}\left(\Lambda\tilde{\Theta}_i\Gamma^{-1}_{\theta}\dot{\tilde{\Theta}}_{i}\right)+2\sum_{j=1}^{N}\text{tr}\left(\Lambda\tilde{\Theta}_j\Gamma^{-1}_{\phi}\dot{\tilde{\Theta}}_{j}\right),
    \label{dlyapg2ia}
\end{align}
grouping the terms,
\begin{align}
    \dot{V}&=\sum_{i=1}^{N}\sum_{j=0}^{N}a_{ij}\left(e^\top_{ij}A^\top_{Hj}P_ie_{ij}+e^\top_{ij}P_iA_{Hj}e_{ij}
    +2\left[e^\top_{ij}P_iB_i\Lambda\tilde{K}_{mi}(\sigma_i-\sigma_j)+\text{tr}\left(\Lambda\tilde{K}_{mi}\Gamma^{-1}_{m}\dot{\tilde{K}}^\top_{mi}\right)\right]
    \right. \nonumber \\
    &\left.+2\left[e^\top_{ij}P_iB_i\Lambda\tilde{K}_{rij}\xi_{ij}+\text{tr}\left(\Lambda\tilde{K}_{rij}\Gamma^{-1}_{r}\dot{\tilde{K}}^\top_{rij}\right)\right]+2\left[e^\top_{ij}P_iB_i\Lambda\tilde{K}_{ij}\sigma_j+\text{tr}\left(\Lambda\tilde{K}_{ij}\Gamma^{-1}_{r}\dot{\tilde{K}}^\top_{ij}\right)\right]\right. \nonumber \\
    &\left.-2\left[e^\top_{ij}P_iB_i\Lambda\tilde{\Theta}_{i}\phi_i(x_i)+\text{tr}\left(\Lambda\tilde{\Theta}_i\Gamma ^{-1}_\theta\dot{\tilde{\Theta}}^\top_i\right)\right]+2\left[e^\top_{ij}P_iB_i\Lambda\tilde{\Theta}_{j}\phi_j(x_j)+\text{tr}\left(\Lambda\tilde{\Theta}_j\Gamma ^{-1}_\phi\dot{\tilde{\Theta}}^\top_j\right)\right]\right),
    \label{dlyapg2ib}
\end{align}
considering as well the trace property of tr$(CD^\top)=D^\top C$, with $C,D\in\mathbb{R}^n$, and the definition of $P_i$, it follows that

\begin{align*}
    &\dot{V}=\sum_{i=1}^{N}\sum_{j=0}^{N}a_{ij}\left(-e^\top_{ij}Q_ie_{ij}
    +2\text{tr}\left(\Lambda\tilde{K}_{mi}(\sigma_i-\sigma_j)e^\top_{ij}P_iB_i+\Lambda\tilde{K}_{mi}\Gamma^{-1}_{m}\dot{\tilde{K}}^\top_{mi}\right) \right.\\
    &\left.+2\text{tr}\left(\Lambda\tilde{K}_{rij}\xi_{ij}e^\top_{ij}P_iB_i+\Lambda\tilde{K}_{rij}\Gamma^{-1}_{r}\dot{\tilde{K}}^\top_{rij}\right)+2\text{tr}\left(\Lambda\tilde{K}_{ij}\sigma_je^\top_{ij}P_iB_i+\Lambda\tilde{K}_{ij}\Gamma^{-1}_{ij}\dot{\tilde{K}}^\top_{ij}\right)\right. \nonumber \\ 
    &\left.-2\text{tr}\left(\Lambda\tilde{\Theta}_i\phi_i(x_i)e^\top_{ij}P_iB_i+\Lambda\tilde{\Theta}_i\Gamma ^{-1}_\theta\dot{\tilde{\Theta}}^\top_i\right)
    +2\text{tr}\left(\Lambda\tilde{\Theta}_j\phi_j(x_j)e^\top_{ij}P_iB_i+\Lambda\tilde{\Theta}_j\Gamma ^{-1}_\phi\dot{\tilde{\Theta}}^\top_j\right)\right),
\end{align*}
factorizing $\Lambda$, we can obtain,
\begin{align*}
    &\dot{V}=\sum_{i=1}^{N}\sum_{j=0}^{N}a_{ij}\left(-e^\top_{ij}Q_ie_{ij}
    +2\Lambda\text{tr}\left(\tilde{K}_{mi}(\sigma_i-\sigma_j)e^\top_{ij}P_iB_i+\tilde{K}_{mi}\Gamma^{-1}_{m}\dot{\tilde{K}}^\top_{mi}\right)\right. \\
    &\left.+2\Lambda\text{tr}\left(\tilde{K}_{rij}\xi_{ij}e^\top_{ij}P_iB_i+\tilde{K}_{rij}\Gamma^{-1}_{r}\dot{\tilde{K}}^\top_{rij}\right)+2\Lambda\text{tr}\left(\tilde{K}_{ij}\sigma_je^\top_{ij}P_iB_i+\tilde{K}_{ij}\Gamma^{-1}_{ij}\dot{\tilde{K}}^\top_{ij}\right) \right. \\
    &\left.-2\Lambda\text{tr}\left(\tilde{\Theta}_i\phi_i(x_i)e^\top_{ij}P_iB_i+\tilde{\Theta}_i\Gamma ^{-1}_\theta\dot{\tilde{\Theta}}^\top_i\right)
    +2\Lambda\text{tr}\left(\tilde{\Theta}_j\phi_j(x_j)e^\top_{ij}P_iB_i+\tilde{\Theta}_j\Gamma ^{-1}_\phi\dot{\tilde{\Theta}}^\top_j\right)\right),
\end{align*}
grouping in terms of the estimators $\tilde{K}_{mi},\tilde{K}_{rij},\tilde{K}_{ij},\tilde{\Theta}_i$,$\tilde{\Theta}_j$, we have
\begin{align*}
    &\dot{V}=\sum_{i=1}^{N}\sum_{j=0}^{N}a_{ij}\left(-e^\top_{ij}Q_ie_{ij}
    +2\Lambda\text{tr}\left(\tilde{K}_{mi}\left((\sigma_i-\sigma_j)e^\top_{ij}P_iB_i+\Gamma^{-1}_{m}\dot{\tilde{K}}^\top_{mi}\right)\right)\right. \\
    &\left.+2\Lambda\text{tr}\left(\tilde{K}_{rij}\left(\xi_{ij}e^\top_{ij}P_iB_i+\Gamma^{-1}_{r}\dot{\tilde{K}}^\top_{rij}\right)\right)+2\Lambda\text{tr}\left(\tilde{K}_{ij}\left(\sigma_je^\top_{ij}P_iB_i+\Gamma^{-1}_{ij}\dot{\tilde{K}}^\top_{ij}\right)\right) \right. \\
    &\left.-2\Lambda\text{tr}\left(\tilde{\Theta}_i\left(\phi_i(x_i)e^\top_{ij}P_iB_i+\Gamma ^{-1}_\theta\dot{\tilde{\Theta}}^\top_i\right)\right)
    +2\Lambda\text{tr}\left(\tilde{\Theta}_j\left(\phi_j(x_j)e^\top_{ij}P_iB_i+\Gamma ^{-1}_\phi\dot{\tilde{\Theta}}^\top_j\right)\right)\right),
\end{align*}
and opening with the adaptive laws~\eqref{al}, we have
\begin{align*}
    \dot{V}&=\sum_{i=1}^{N}\left(-\sum_{j=0}^{N}a_{ij}e^\top_{ij}Q_ie_{ij}
    +2\Lambda\sum_{j=0}^{N}a_{ij}\text{tr}\left(\tilde{K}_{mi}\left((\sigma_i-\sigma_j)e^\top_{ij}P_iB_i-\sum_{\hat{j}=0}^{N}\left(\sigma_i-\sigma_{\hat{j}}\right)e^\top_{i\hat{j}}P_i{B_i}\right)\right)\right.\\
    &+2\Lambda\sum_{j=0}^{N}a_{ij}\text{tr}\left(\tilde{K}_{rij}\left(\xi_{ij}e^\top_{ij}P_iB_i-\xi_{ij} e^\top_{ij}P_i {B_i}\right)\right) 
    +2\Lambda\sum_{j=0}^{N}a_{ij}\text{tr}\left(\tilde{K}_{ij}\left(\sigma_je^\top_{ij}P_iB_i-\sigma_j(x_j)e^\top_{ij}P_i {B_i}\right)\right) \\
    &\left.-2\Lambda\sum_{j=0}^{N}a_{ij}\text{tr}\left(\tilde{\Theta}_i\left(\phi_i(x_i)e^\top_{ij}P_iB_i-\phi_j(x_j) e^\top_{ij}P_i {B_i}\right)\right)
    +2\Lambda\sum_{j=0}^{N}a_{ij}\text{tr}\left(\tilde{\Theta}_j\left(\phi_j(x_j)e^\top_{ij}P_iB_i-\phi_i(x_i) e^\top_{ij}P_i {B_i}\right)\right)\right),
\end{align*}

Because $K_{mi},K_{rij},K_{ij},\Theta_i,\Theta_j$ are constants, therefore $\dot{\tilde{K}}_{mi}=\dot{K}_{mi}$, $\dot{\tilde{K}}_{rij}=\dot{K}_{rij}$, $\dot{\tilde{K}}_{ij}=\dot{K}_{ij}$, $\dot{\tilde{\Theta}}_i=\dot{\Theta}_i$ and $\dot{\tilde{\Theta}}_j=\dot{\Theta}_j$, then we can reduce to
\begin{align*}
    \dot{V}=\sum_{i=1}^{N}\sum_{j=0}^{N}a_{ij}(-e^\top_{ij}Q_ie_{ij})
    \leq\sum_{i=1}^{N}-\lambda_{min}(Q)\sum_{j=0}^{N}a_{ij}\norm{e_{ij}}^2\leq 0,
\end{align*}
where using Barbalat's lemma~\cite{farkas2016variations} and with definition~\ref{uub}, the synchronization error is UUB with \eqref{lyapunov_nl} as a valid Lyapunov function.

\end{proof}

We can now state the main stability result of this synchronization problem in the following theorem. \\ 

\begin{theorem} \label{thm1}
Let Assumptions \ref{assum:graphs}-\ref{asuum:one_agent} hold. The dynamics generated by the set of agents in~\eqref{nonlinear-system-imu}, with control law~\eqref{u1_nl_imu2} for the leader agent, and control law~\eqref{u_th3} for the followers, guarantee UUB for all initial conditions in the synchronization, i.e., $\lim_{t \to \infty} \norm{e_{ij}(t)}=0$ and $\lim_{t \to \infty} \norm{e_{1}(t)}=0$, with $\norm{x_j(t)}<M_{xj}$ $\forall t \in [0, T]$ for a constant $M_{xj} >0$.
\end{theorem}

\begin{proof}
From the hypothesis we know that the reference signal $x_m(t)$ are bounded, from Lemma~\ref{lem1} it follows that the synchronization error $e_{ij}$ and constants $K_{mi},K_{ij},K_{rij},\Theta_i,\Theta_j$ are UUB. The dynamics of the reference and the states are then also bounded, i.e., $x_j,\dot{x}_j,x_m,\dot{x}_m$ are bounded. Thus, $x_i(t)=e_{ij}+x_j(t)$ is UUB, and at the same time, it implies that $u_i(t)$ is bounded as well as $\dot{x}_i$ and $\dot{e}_{ij}$. To ensure uniform continuity of the Lyapunov function derivative \eqref{lyapdot}, its second derivative is
\begin{equation*}
    \Ddot{V}=-2\sum_{i=1}^{N}\sum_{j=0}^{N}a_{ij}e^\top_{ij}Q_ie_{ij},
\end{equation*}
and is bounded because $V(t)\geq 0$ and $\dot{V}(t)\leq 0$. Thus, from Barbalat's Lemma, we have that $\lim_{t\to\infty}\dot{V}(t)=0$. Therefore, we can conclude that  $\lim_{t\to\infty}\norm{e_{ij}(t)}$ is UUB.
\end{proof}

In the case of linear agents with $\omega_i=0$, the distributed control law used for synchronizing agents that do not communicate with the reference is
\begin{equation}
    u_i=\sum_{j=1}^{N}a_{ij}{{K}_{ij}}x_j+K_{mi}\Xi_{i}+\sum_{j=1}^{N}a_{ij}K_{rij}\xi_{ij},
    \label{u_th}
\end{equation}

Similarly as in~\eqref{eq6}, for follower agents $i \in [2,\hdots,N]$. Then, the following corollary presents the stability result for this distributed case. \\

\begin{corollary} \label{lem1}
Let Assumptions \ref{assum:graphs}-\ref{asuum:one_agent} hold. Consider a linear system $\dot{x}_i=A_ix_i+B_iu_i$ with a reference system \eqref{nl_ref}, and employing the control and adaptive laws \eqref{u_th}--\eqref{al}. Then, the function 
\begin{equation}
    V =\sum_{i=1}^{N}\sum_{j=0}^{N}a_{ij}e_{ij}^\top P_ie_{ij}+ \sum_{i=1}^{N} \text{tr}\left (\Lambda\tilde{K}_{mi}\Gamma_{m}\tilde{K}^\top_{mi} \right )+ \sum_{i=1}^{N} \sum_{j=1}^{N}a_{ij} \text{tr}\left (\Lambda{\tilde{K}_{ij}}\Gamma_{ij}\tilde{K}^\top_{ij} \right ) + \sum_{i=1}^{N}\sum_{j=1}^{N}a_{ij}\text{tr}\left (\Lambda{\tilde{K}_{rij}}\Gamma_{r}\tilde{K}^\top_{rij} \right )
    \label{u_thxxx}
\end{equation}
is a valid Lyapunov function.
\end{corollary} 

\begin{proof}
It follows the same procedure as Lemma~\ref{lem_nl} with the error dynamics obtained as

By the definition of $\xi_{ij}$ to expand $u_j$, we obtain
\begin{align*}
    \sum_{j=1}^{N}a_{ij}\dot{e}_{ij}=\sum_{j=0}^{N}a_{ij}A_{Hj}e_{ij}
    +\sum_{j=1}^{N}a_{ij}B_i\Lambda(\tilde{K}_{mi}\left(x_i-x_j\right)+\tilde{K}_{ij}x_j+\tilde{K}_{rij}\xi_{ij}).
\end{align*}

to implies that the synchronization error is bounded by~\eqref{u_thxxx}, in a procedure similar to that of Lemma~\ref{lem_nl} and Theorem~\ref{thm1}.
\end{proof}

From this analysis, it is possible to prove that synchronization error is uniformly bounded. With this information, we can state the case with input magnitude saturation for the previous development techniques in the next section.

\section{Input Magnitude Saturation Adaptive Control with Reinforcement Learning} \label{MSAC}
This section presents the main result of the work as an additional case with the heterogeneous synchronization of agents with uncertainties and in the presence of input saturation. The input saturation for an agent connected directly with a reference is handled as
\begin{equation}
    u_{i,sat}(t)=u_{\text{max}}\text{sat}\left(\frac{u_{i}(t)}{u_{i,\text{max}}}\right),
    \label{msac_constraint}
\end{equation}
where the MRAC-RL controller output is $u_{i}(t)$. The saturation function \( \text{sat}(x) \) limits the value of \( x \) to lie within a specified range, such that 
\[
\text{sat}(x) = 
\begin{cases} 
\text{max\_val} & \text{if } x > \text{max\_val} \\
x & \text{if } \text{min\_val} \le x \le \text{max\_val} \\
\text{min\_val} & \text{if } x < \text{min\_val}
\end{cases}
\]
where \( \text{min\_val} \) and \( \text{max\_val} \) are the lower and upper bounds, respectively. This saturation may incur a disturbance in the controller action, defined as
\begin{equation}
    \Delta u_i(t) =u_{i}(t)-u_{i,sat}(t).
\end{equation}

It is easy to see that $\Delta u_i(t) = 0$ when the desired control $u_{i,ac}(t)$ does not saturate, which analytically leads to the definition of a performance error $e_{pij}$ whose dynamics is represented as
\begin{equation}
    \sum_{j=1}^{N}a_{ij}\dot{e}_{pij}=\sum_{j=0}^{N}a_{ij}A_{Hj}e_{pij}+\sum_{j=1}^{N}a_{ij}B_iK^\top_{pi}\Delta u_i,
\end{equation}
We introduce then a new performance error $e_{uij}=e_{ij}-e_{pij}$, which consider the disturbance presented by the variation $\Delta u_i(t)$ as
\begin{align}
    \sum_{j=1}^{N}a_{ij}\dot{e}_{uij}&=\sum_{j=0}^{N}a_{ij}A_{Hj}e_{ij}
    +\sum_{j=1}^{N}a_{ij}B_i\Lambda(\tilde{K}_{mi}\left(\sigma_i-\sigma_j\right)+\tilde{K}_{ij}\sigma_j(x_j)+\tilde{K}_{rij}\xi_{ij}
    +\tilde{\Theta}_j\phi_{j}(x_j)-\tilde{\Theta}_i\phi_i(x_i)) \nonumber \\
    &-\sum_{j=0}^{N}a_{ij}A_{Hj}e_{pij}-\sum_{j=1}^{N}a_{ij}B_iK^\top_{pi}\Delta u_i,
\end{align}
grouping by $A_{Hj}$, we have
\begin{align}
    \sum_{j=1}^{N}a_{ij}\dot{e}_{uij}&=\sum_{j=1}^{N}a_{ij}A_{Hj}(e_{ij}-e_{pij})
    +\sum_{j=1}^{N}a_{ij}B_i\Lambda(\tilde{K}_{mi}\left(\sigma_i-\sigma_j\right)+\tilde{K}_{ij}\sigma_j(x_j)+\tilde{K}_{rij}\xi_{ij}
    +\tilde{\Theta}_j\phi_{j}(x_j) \nonumber \\
    &-\tilde{\Theta}_i\phi_i(x_i)-K_{pi}\Delta u_i).
    \label{msac_error}
\end{align}

This suggests a modification to the adaptive laws:
\begin{subequations} \label{al_msac}
    \begin{align}
        {\dot{K}_{ij}}=&-\Gamma_{ij}\sigma_j(x_j)e^\top_{uij}P_i {B_i}, \\
        {\dot{K}_{pi}}=&-\Gamma_{p}\Delta u_ie^\top_{uij}P_i {B_i}, \\
        {\dot{K}_{mi}}=&-\Gamma_m\Xi_{i}e^\top_{uij}P_i{B_i},\\
        \dot{K}_{rij}=&-\Gamma_{r}\xi_{ij} e^\top_{uij}P_i {B_i}, \\
        \dot{\Theta}_{j}=&-\Gamma_{\phi}\phi_j(x_j) e^\top_{uij}P_i {B_i}, \\
        \dot{\Theta}_{i}=&-\Gamma_{\theta}\phi_i(x_i) e^\top_{uij}P_i {B_i},
    \end{align}
\end{subequations}

in which another positive definite gain matrix $\Gamma_{p}\succ 0$ has been introduced. We can define the following proposition.

\begin{proposition} \label{lem2}
Let Assumptions~\ref{assum:graphs}-\ref{asuum:one_agent} hold. Consider a network of systems~\eqref{nonlinear-system-imu}, control and adaptive laws \eqref{u_th3}--\eqref{al_msac}, the input magnitude constraint~\eqref{msac_constraint} . Then, the synchronization error~\eqref{msac_error} is UUB for all initial conditions
\end{proposition} 

\begin{proof}
A Lyapunov function is proposed as
\begin{align}
    V &=\sum_{i=1}^{N}\sum_{j=1}^{N}a_{ij}e_{uij}^\top P_ie_{uij}+ \sum_{i=1}^{N} \text{tr}\left (\Lambda\tilde{K}_{mi}\Gamma_{m}\tilde{K}^\top_{mi} \right )+ \sum_{i=1}^{N} \sum_{j=1}^{N}a_{ij} \text{tr}\left (\Lambda{\tilde{K}_{ij}}\Gamma_{ij}\tilde{K}^\top_{ij} \right ) + \sum_{i=1}^{N}\sum_{j=1}^{N}a_{ij}\text{tr}\left (\Lambda{\tilde{K}_{rij}}\Gamma_{r}\tilde{K}^\top_{rij} \right ) \nonumber \\
    &+ \sum_{i=1}^{N} \sum_{j=1}^{N}a_{ij}\; \text{tr}\left (\Lambda{\tilde{\Theta}_{j}}\Gamma_{\phi}\tilde{\Theta}^\top_{j} \right )+\sum_{i=1}^{N}\text{tr}\left(\Lambda\tilde{\Theta}_i\Gamma^{-1}_{\Theta}\tilde{\Theta}^\top_i\right),
    \label{lyapunov_nl_msac}
\end{align}
The time derivative is
\begin{align}
    \dot{V}&=\sum_{i=1}^{N}\sum_{j=0}^{N}\dot{e}_{uij}^\top P_ie_{uij}+\sum_{i=1}^{N}\sum_{j=0}^{N}e_{ij}^\top P_i\dot{e}_{uij}+2\sum_{i=1}^{N}\sum_{j=1}^{N}a_{ij}\text{tr}\left(\Lambda\tilde{K}^\top_{ij}\Gamma^{-1}_{ij}\dot{\tilde{K}}_{ij}\right)+2\sum_{i=1}^{N}\text{tr}\left(\Lambda\tilde{K}^\top_{mi}\Gamma^{-1}_{m}\dot{\tilde{K}}_{mi}\right) \nonumber \\
    &+2\sum_{i=1}^{N}\sum_{j=1}^{N}a_{ij}\text{tr}\left(\Lambda\tilde{K}_{rij}\Gamma^{-1}_{m}\dot{\tilde{K}}_{rij}\right) \nonumber \\
    &-2\sum_{i=1}^{N}\text{tr}\left(\Lambda\tilde{\Theta}_i\Gamma^{-1}_{\theta}\dot{\tilde{\Theta}}_{i}\right)+2\sum_{j=1}^{N}\text{tr}\left(\Lambda\tilde{\Theta}_j\Gamma^{-1}_{\phi}\dot{\tilde{\Theta}}_{j}\right)
    \label{lyapdot_msac}
\end{align}
which expanded through the definition of the error dynamics, is
\begin{align}
    \dot{V}&=\sum_{i=1}^{N}\left(\sum_{j=0}^{N}a_{ij}A_{Hj}(e_{ij}-e_{pij})+\sum_{j=0}^{N}a_{ij}B_i\Lambda(\tilde{K}_{mi}\left(\sigma_i-\sigma_j\right)+\tilde{K}_{ij}\sigma_j(x_j)+\tilde{K}_{rij}\xi_{ij}+\tilde{\Theta}_j\phi_{j}(x_j)\right. \nonumber \\
    &\left.-\tilde{\Theta}_i\phi_i(x_i)-K_{pi}\Delta u_i)\right)^\top P_ie_{uij} +\sum_{i=1}^{N}e_{uij}^\top P_i\left(\sum_{j=0}^{N}a_{ij}A_{Hj}(e_{ij}-e_{pij})+\sum_{j=0}^{N}a_{ij}B_i\Lambda(\tilde{K}_{mi}\left(\sigma_i-\sigma_j\right) \right. \nonumber \\
    &\left.+\tilde{K}_{ij}\sigma_j(x_j)+\tilde{K}_{rij}\xi_{ij}+\tilde{\Theta}_j\phi_{j}(x_j)-\tilde{\Theta}_i\phi_i(x_i)-K_{pi}\Delta u_i)\right)+2\sum_{i=1}^{N}\sum_{j=1}^{N}a_{ij}\text{tr}\left(\Lambda\tilde{K}^\top_{ij}\Gamma^{-1}_{ij}\dot{\tilde{K}}_{ij}\right) \nonumber \\
    &+2\sum_{i=1}^{N}\text{tr}\left(\Lambda\tilde{K}^\top_{mi}\Gamma^{-1}_{m}\dot{\tilde{K}}_{mi}\right)+2\sum_{i=1}^{N}\sum_{j=1}^{N}\text{tr}\left(\Lambda\tilde{K}_{rij}\Gamma^{-1}_{m}\dot{\tilde{K}}_{rij}\right)-2\sum_{i=1}^{N}\text{tr}\left(\Lambda\tilde{\Theta}_i\Gamma^{-1}_{\theta}\dot{\tilde{\Theta}}_{i}\right)+2\sum_{j=1}^{N}\text{tr}\left(\Lambda\tilde{\Theta}_j\Gamma^{-1}_{\phi}\dot{\tilde{\Theta}}_{j}\right)
    \label{dlyapg2ia_msac}
\end{align}
grouping the terms, and with the definition of $e_{uij}$, then we have
\begin{align}
    \dot{V}&=\sum_{i=1}^{N}\left(\sum_{j=0}^{N}a_{ij}A_{Hj}e_{uij}+\sum_{j=0}^{N}a_{ij}B_i\Lambda(\tilde{K}_{mi}\left(\sigma_i-\sigma_j\right)+\tilde{K}_{ij}\sigma_j(x_j)+\tilde{K}_{rij}\xi_{ij}+\tilde{\Theta}_j\phi_{j}(x_j)\right. \nonumber \\
    &\left.-\tilde{\Theta}_i\phi_i(x_i)-K_{pi}\Delta u_i)\right)^\top P_ie_{uij}+\sum_{i=1}^{N}e_{uij}^\top P_i\left(\sum_{j=0}^{N}a_{ij}A_{Hj}e_{uij}+\sum_{j=0}^{N}a_{ij}B_i\Lambda(\tilde{K}_{mi}\left(\sigma_i-\sigma_j\right)\right. \nonumber \\
    &\left.+\tilde{K}_{ij}\sigma_j(x_j)+\tilde{K}_{rij}\xi_{ij}+\tilde{\Theta}_j\phi_{j}(x_j)-\tilde{\Theta}_i\phi_i(x_i)-K_{pi}\Delta u_i)\right)+2\sum_{i=1}^{N}\sum_{j=1}^{N}a_{ij}\text{tr}\left(\Lambda\tilde{K}^\top_{ij}\Gamma^{-1}_{ij}\dot{\tilde{K}}_{ij}\right) \nonumber \\
    &+2\sum_{i=1}^{N}\text{tr}\left(\Lambda\tilde{K}^\top_{mi}\Gamma^{-1}_{m}\dot{\tilde{K}}_{mi}\right)+2\sum_{i=1}^{N}\sum_{j=1}^{N}\text{tr}\left(\Lambda\tilde{K}_{rij}\Gamma^{-1}_{m}\dot{\tilde{K}}_{rij}\right)-2\sum_{i=1}^{N}\text{tr}\left(\Lambda\tilde{\Theta}_i\Gamma^{-1}_{\theta}\dot{\tilde{\Theta}}_{i}\right) \nonumber \\
    &+2\sum_{j=1}^{N}\text{tr}\left(\Lambda\tilde{\Theta}_j\Gamma^{-1}_{\phi}\dot{\tilde{\Theta}}_{j}\right).
    \label{dlyapg2ib_msac}
\end{align}
Taking as reference the Lemma~\eqref{lem_nl}, it can be concluded that $e_{uij},\tilde{K}_{ij},\tilde{K}_{mi},\tilde{K}_{rij},\tilde{\Theta}_i,\tilde{\Theta}_j$ are bounded. Since all controller parameters are bounded, a bounded input to the reference model implies that the states $x_i$ are bounded. Therefore, synchronization errors $e_{uij}$ are bounded. Thus, in a similar fashion to the proof of Theorem~\ref{thm1} from Barbalat's Lemma, we have that $\lim_{t\to\infty}\dot{V}(t)=0$. Therefore, we can conclude that  $\lim_{t\to\infty}\norm{e_{ij}(t)}=0$ the synchronization error tends to zero globally, asymptotically, and uniformly.
\end{proof}

Proposition~\ref {lem2} allows the heterogeneous synchronization of agents to improve the performance of reinforcement learning techniques through adaptive techniques in scenarios of heterogeneity, uncertainty, and saturation. With this information, we present the simulation results obtained in the next section.

\section{Numerical Analysis} \label{sims}
In this section two cases of experimental analysis are presented. Initially, a network pendulum model for the validation of adaptive control algorithms with reinforcement learning. Following this, the validation of the algorithms for saturation management is presented.

\subsection{Network of pendulum systems for validation of adaptive control}

Consider the following nonlinear model of an inverted pendulum
\begin{equation}
    ml^2\ddot{\theta}=mgl\sin\theta-b\dot{\theta}+\tau,
    \label{nonlinear-model}
\end{equation}
where $m$ is the pendulum mass, $g$ is the gravitational constant, $l$ is the length pendulum, and $\tau$ is the force provided to the system. The goal is to maintain a non-zero set-point for the states $\theta,\dot{\theta}$. For consistency with the rest of the paper, we denote $x=[\theta,\dot{\theta}]$. We use an off-the-shelf \textit{Deep Deterministic Policy Gradient Agent} pre-trained policy from MATLAB\textregistered. This policy was trained to swing up and balance an inverted pendulum. Training process details can be found in Mathworks
~\cite{Mathworks}. 

Initially, we present the results of the systems implementing only the RL algorithm and compare them with the distributed MRAC strategy. Moreover, we compare the cases with and without input-matched uncertainties. For the training procedure, the system parameters $m=l=1$, $b=0$, and $g=9.81$ are selected. \\

The communications graph used for the test network is shown in Figure~\ref{fig:graph}. The response of the network, using only the reinforcement learning strategy in all the agents, is observed in Figure~\ref{fig:system1}. As expected, with no input-matched uncertainty and homogeneous agents identical to the reference model, the RL-trained policy stabilizes the network of agents. We are showing the trajectories of all agents, but the fast synchronization and stabilization make all the plots overlap into a single line. 

\begin{figure}[ht]
    \centering
    \includegraphics[width=0.35\textwidth]{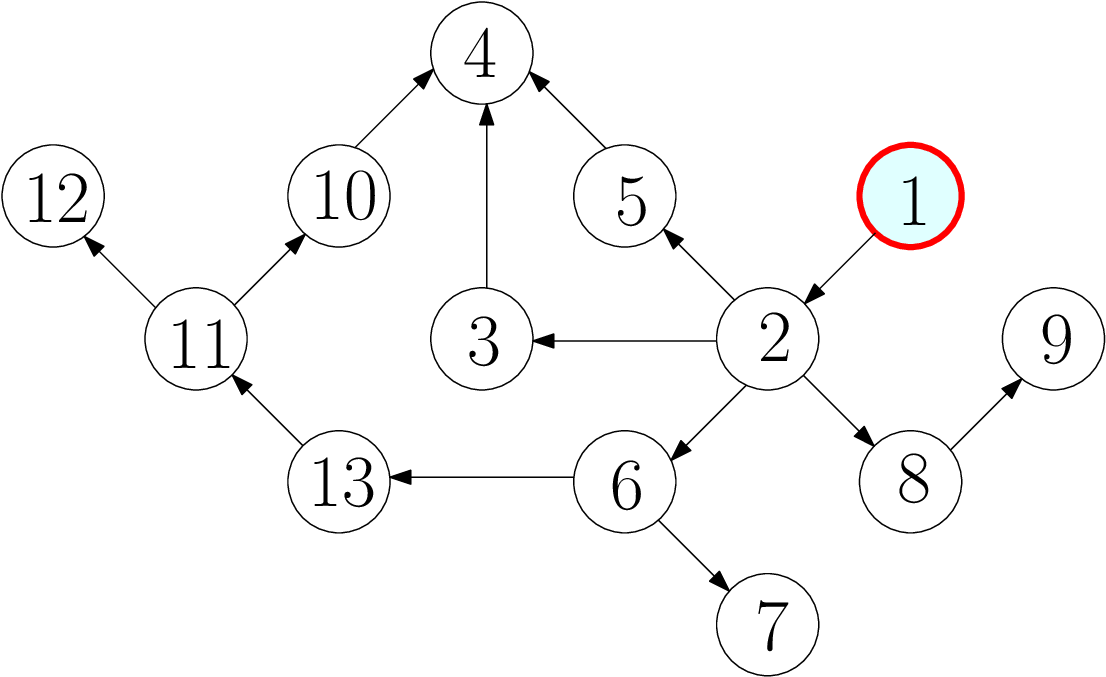}
    \caption{Distributed communication network, represented as a directed graph. The red circle indicates the leader agent. Each agent only has communication with the agents in its neighborhood according to the specified topology.}
    \label{fig:graph}
\end{figure}

\begin{figure}[ht]
    \centering
    \includegraphics[width=0.48\textwidth]{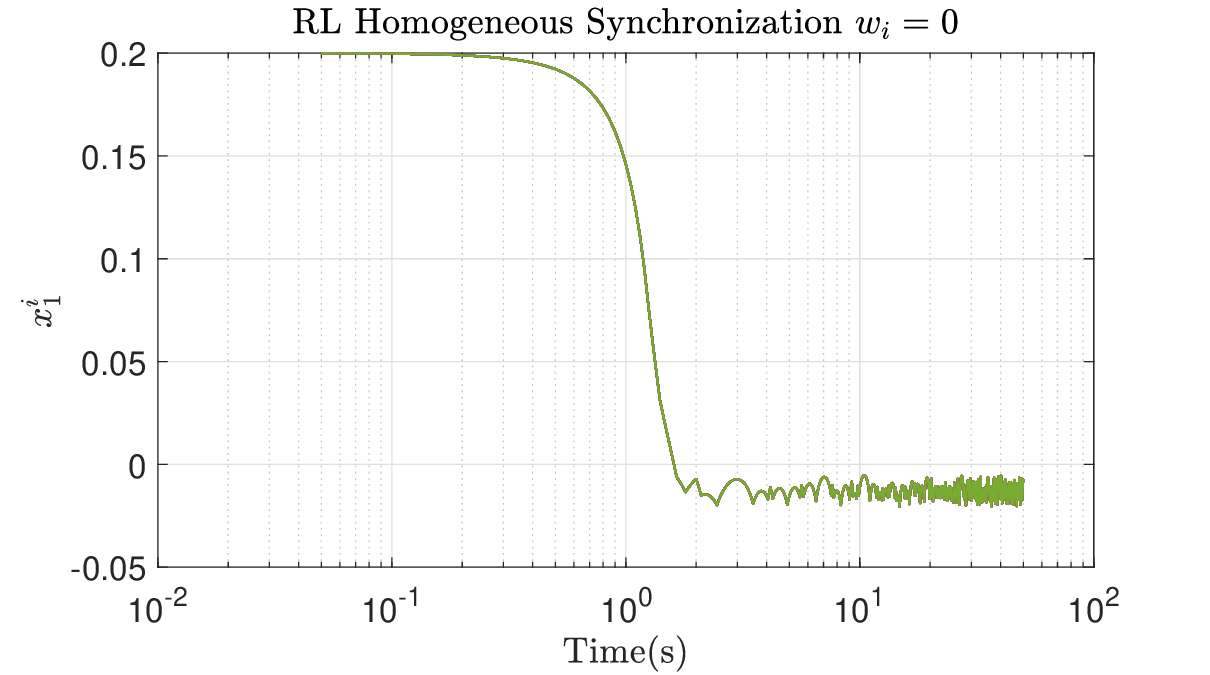}
    \caption{Synchronization of homogeneous agents with Reinforcement Learning technique. The algorithm policy was trained offline with respect to the reference.}
    \label{fig:system1}
\end{figure}

Figure~\ref{fig:system4} shows the same experiment as in Figure~\ref{fig:system1}, now including an input-matched uncertainty to the entire network of $w_i=0.1\sin (t)$ and the DMRAC-RL strategy. Specifically, this shows that when the reference model matches the model of the agents, the pre-trained RL policy alongside the DMRAC-RL stabilizes the nonlinear pendulums. Please note that the figures omit a legend due to space constraints, given the large number of lines representing the trajectories of all agents in the network.

\begin{figure}[ht]
    \centering
    \includegraphics[width=0.48\textwidth]{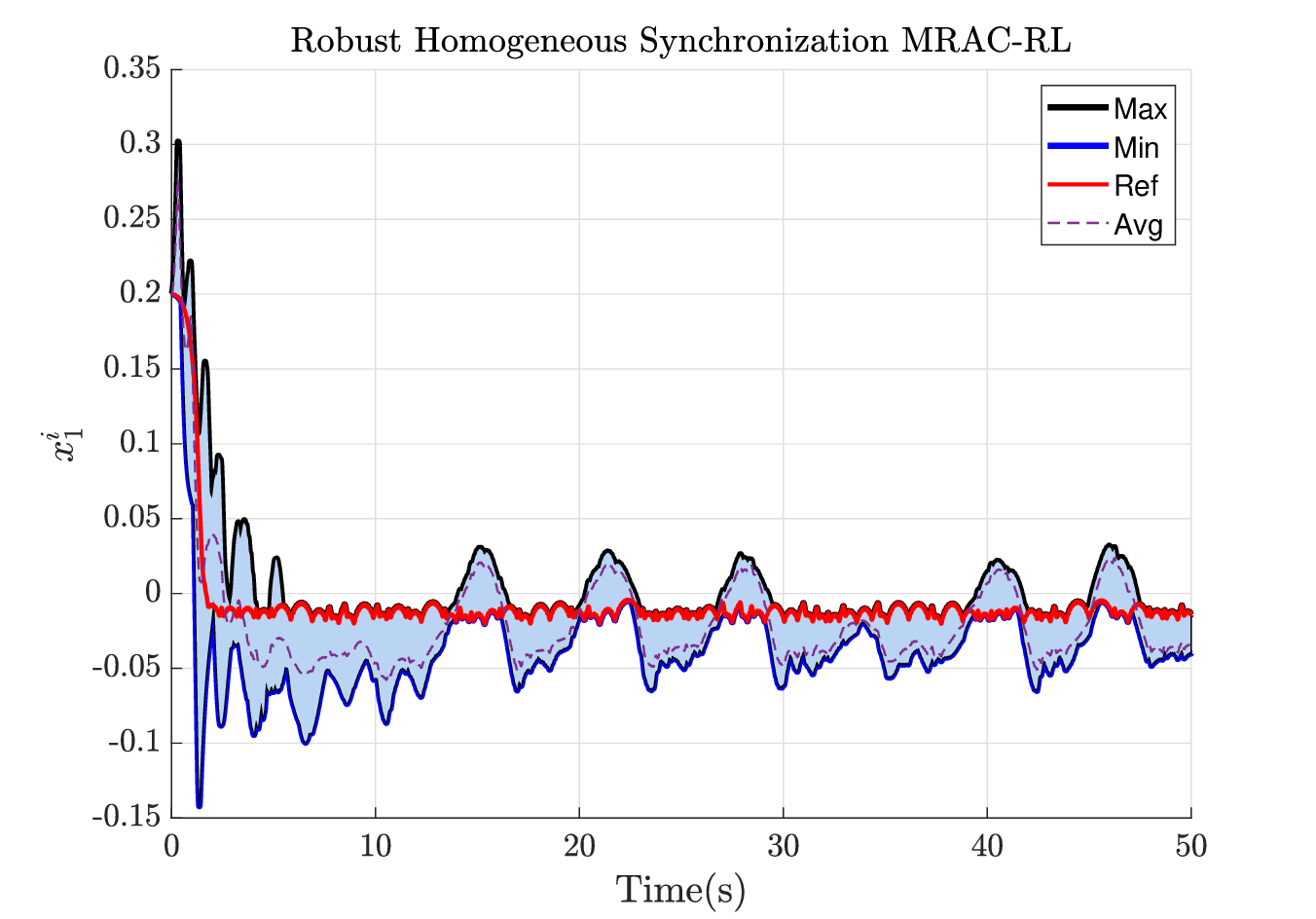}
    \caption{MRAC-RL homogeneous synchronization with input matched uncertainties. The worst agents' response delimits the shaded area, the average response is dotted, and the reference is red.}
    \label{fig:system4}
\end{figure}

Figure~\ref{fig:system5} shows the response of the nonlinear inverted pendulum network with variations of the model parameters $l$ and $m$ uniformly sampled from $\left[0.75,1.25\right]$. However, contrary to previous results, some nodes are unstable, and their states diverge.

\begin{figure}[ht]
    \centering
    \includegraphics[width=0.48\textwidth]{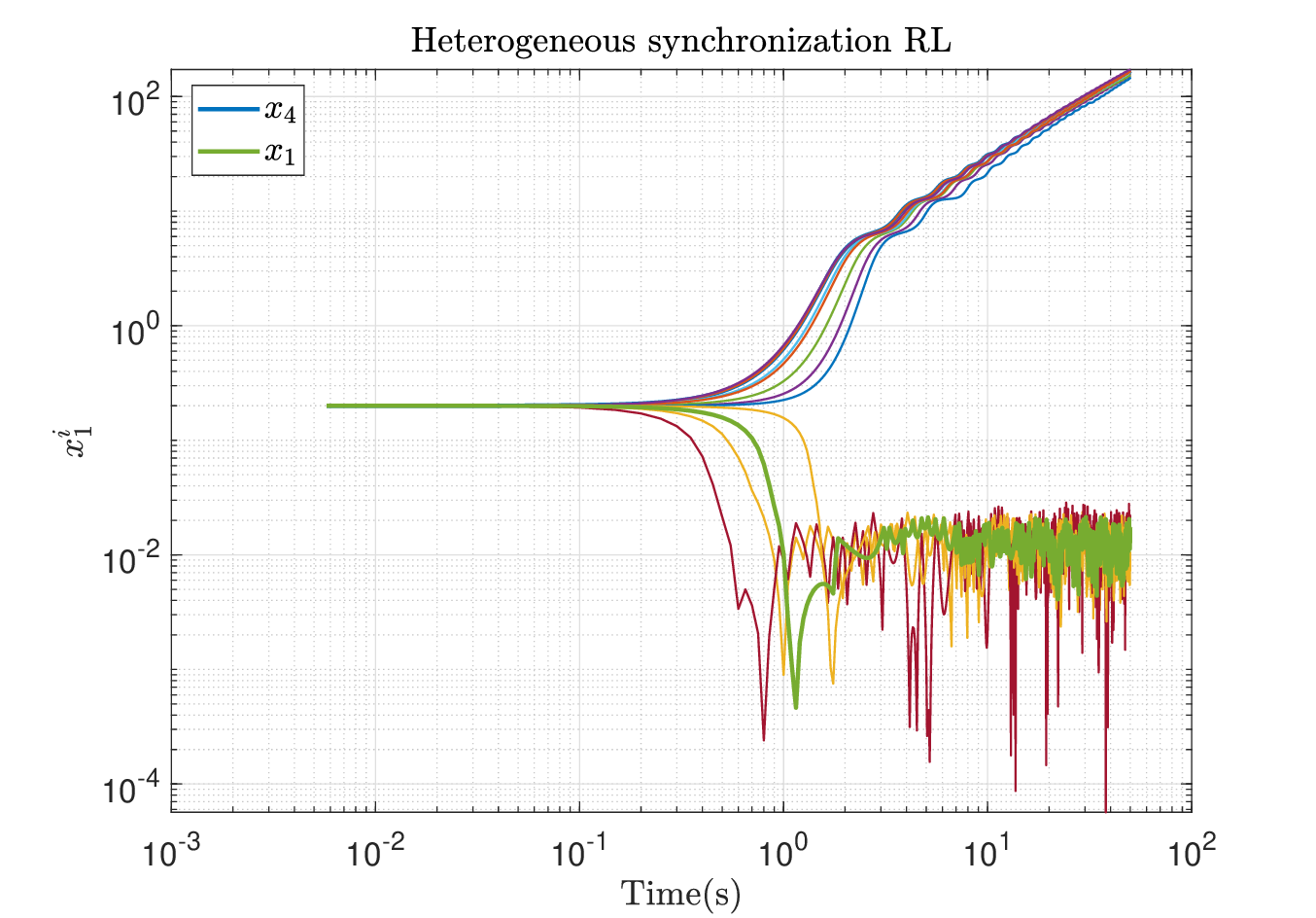}
    \caption{RL heterogeneous synchronization with input matched uncertainties. The response of the states of each of the agents in the network is shown. The graph of $x_1$ shows one of the agents whose dynamics converge and with $x_4$ an agent whose dynamics diverge given the alteration in the model parameters.} 
    \label{fig:system5}
\end{figure}

The DMRAC-RL control law is included to counteract these uncertainties. Figure~\ref{fig:system6} shows the synchronization response of the nonlinear distributed system with heterogeneous agents. The agent parameters are uniformly sampled from $\left[0.75,1.25\right]$ and the input matched uncertainty is $w_i=0.1\sin (t)$; the graph shows the worst results above and below for each agent, with the dotted line showing the average value of the agents at each timestep and the reference in red. Note that even under these adverse conditions, the system synchronizes. It is important to highlight that this is the main contribution of this work. With the variations in the agent's parameters concerning the reference model, the response of the policy trained on the reference model is not robust, as shown in Figure~\ref{fig:system5}, whereas the proposed DMRAC-RL strategy allows synchronization.

\begin{figure}[ht]
    \centering
    \includegraphics[width=0.48\textwidth]{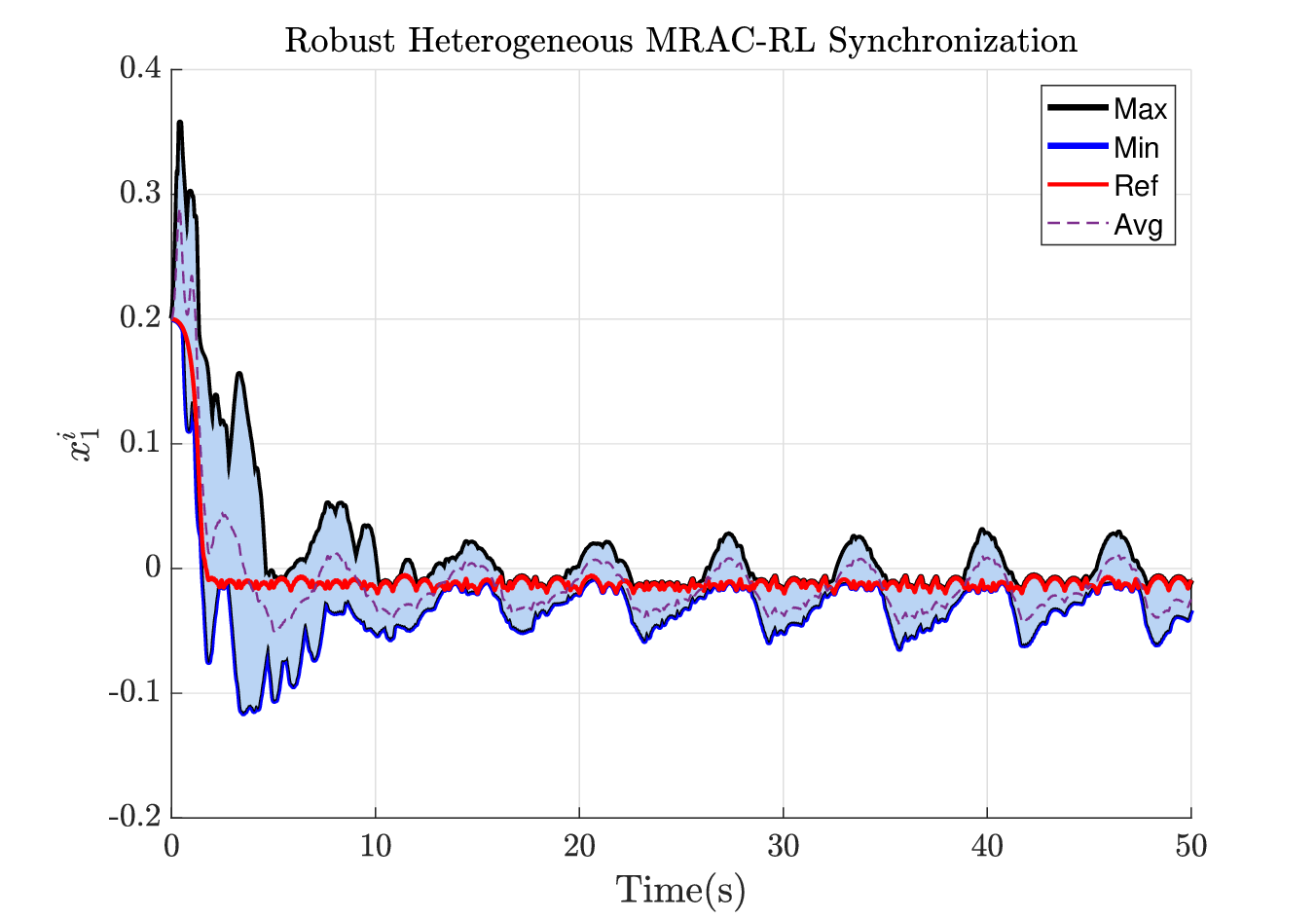}
    \caption{MRAC-RL heterogeneous synchronization with input matched uncertainties to validate the synchronization of the developed technique. The worst agents' response delimits the shaded area, the average response is dotted, and the reference is red.}
    \label{fig:system6}
\end{figure}

Next, we show the performance of the proposed DMRAC-RL framework on the network by including a random input-matched uncertainty with uniform distribution sampled along $\left[-1,1\right]$.  Figure~\ref{fig:systemRnd} shows the trajectories generated by the network of pendulums with these uncertainties. The network of heterogeneous nonlinear systems with random input-matched uncertainty synchronizes. 

Finally, in Figure~\ref{fig:systemTrack}, we show the performance of the proposed method on a tracking problem of a multi-step reference signal. Recall that only the leader agent can access the reference model and the policy trained through the RL algorithm. Still, the network tracks the reference signal with heterogeneous parameters, input matched uncertainties, and initial conditions variations.

\begin{figure}[ht]
    \centering
    \includegraphics[width=0.48\textwidth]{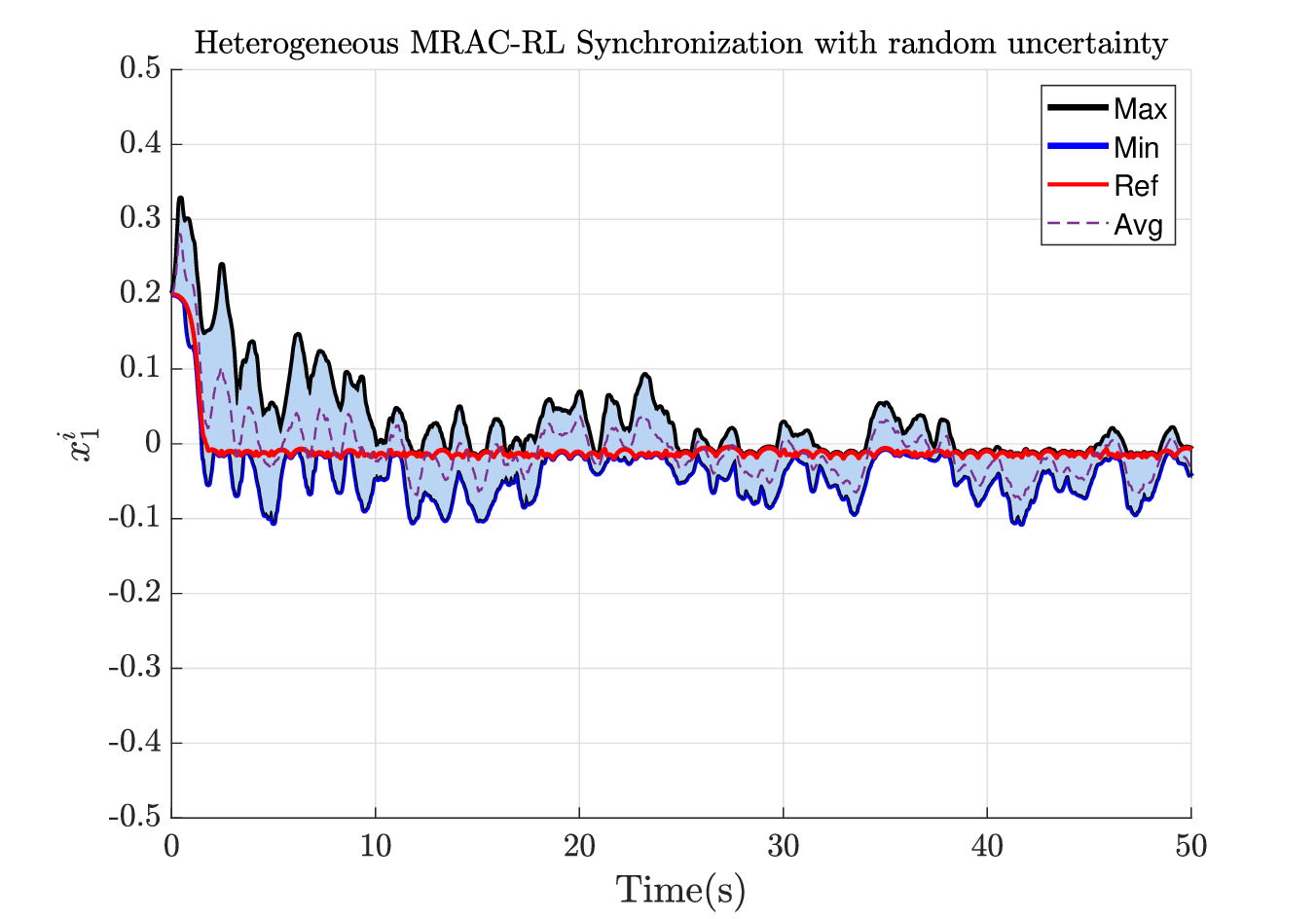}
    \caption{MRAC-RL heterogeneous synchronization with random input matched uncertainties. The worst agents' response delimits the shaded area, the average response is dotted, and the reference is red.}
    \label{fig:systemRnd}
    
\end{figure}

\begin{figure}[ht]
    \centering
    \includegraphics[width=0.48\textwidth]{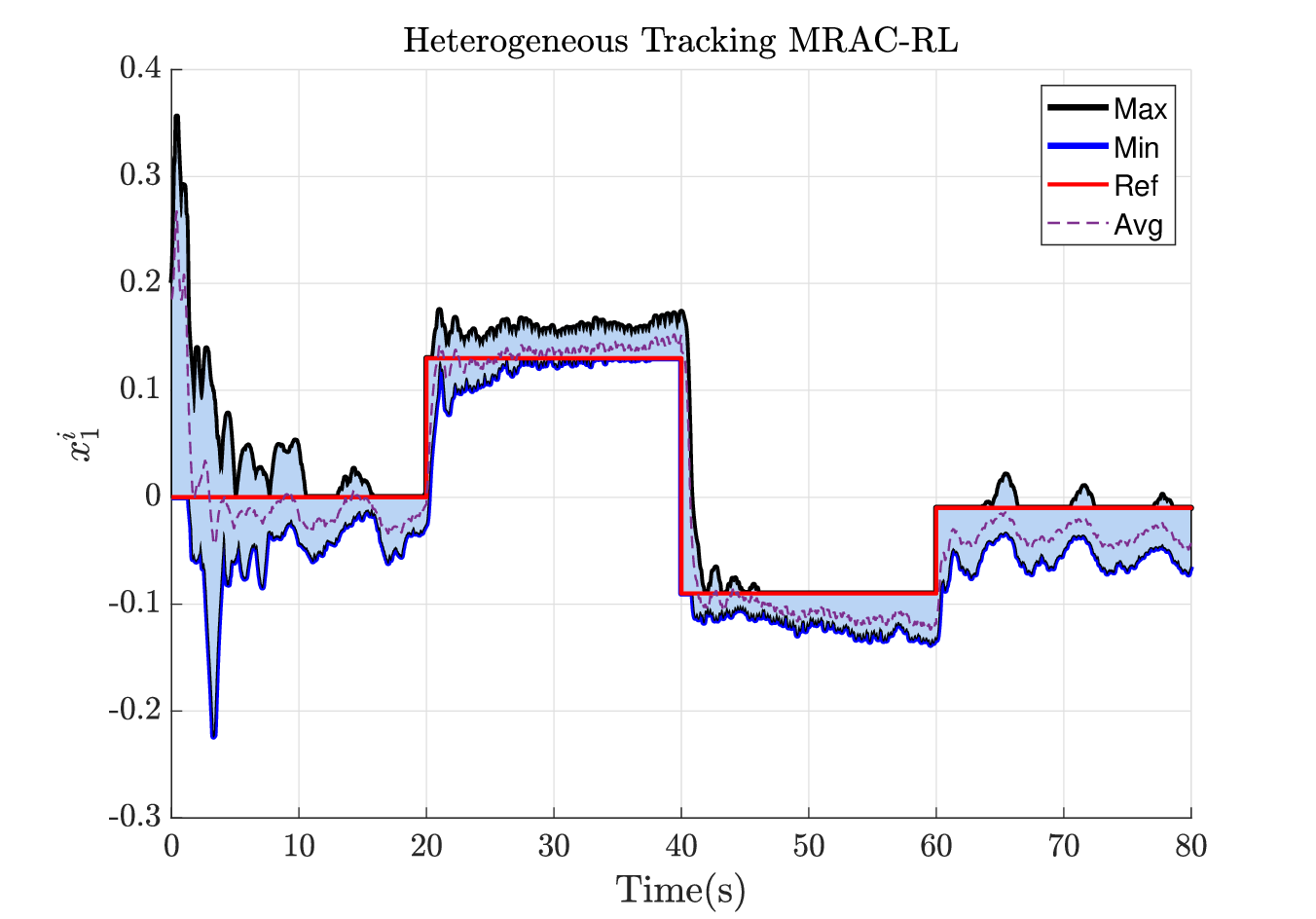}
    \caption{MRAC-RL heterogeneous tracking with input matched uncertainties. The worst agents' response delimits the shaded area, the average response is dotted, and the reference is red.}
    \label{fig:systemTrack}
\end{figure}

\subsection{Dynamic model for magnitude saturation validation}

Next, we show the performance of the proposed DMSAC-RL framework on a MIMO linear model in the form
\begin{eqnarray}
 \dot{x}_i = \left[ \begin{array}{c}
x_2 \\
x_3+w_2 \\
-x_1-2x_2-3x_3+u_1\end{array} \right]\
\label{mimosystem}
\end{eqnarray}
The goal is to maintain a non-zero set point for the three-state system. 

Figure~\ref{fig:DMSAC_RL} shows the trajectories generated by the network of MIMO systems with the input Magnitude Saturation Adaptive Control, validating that it is possible to perform a heterogeneous synchronization of multiple input systems with saturation management. Recall that only one agent communicates directly with the reference agent trained through the RL algorithm. Still, the network is synchronized with heterogeneous parameters and variations in its initial conditions and different system and network configurations.
\begin{figure}[ht]
    \centering
    \includegraphics[width=0.48\textwidth]{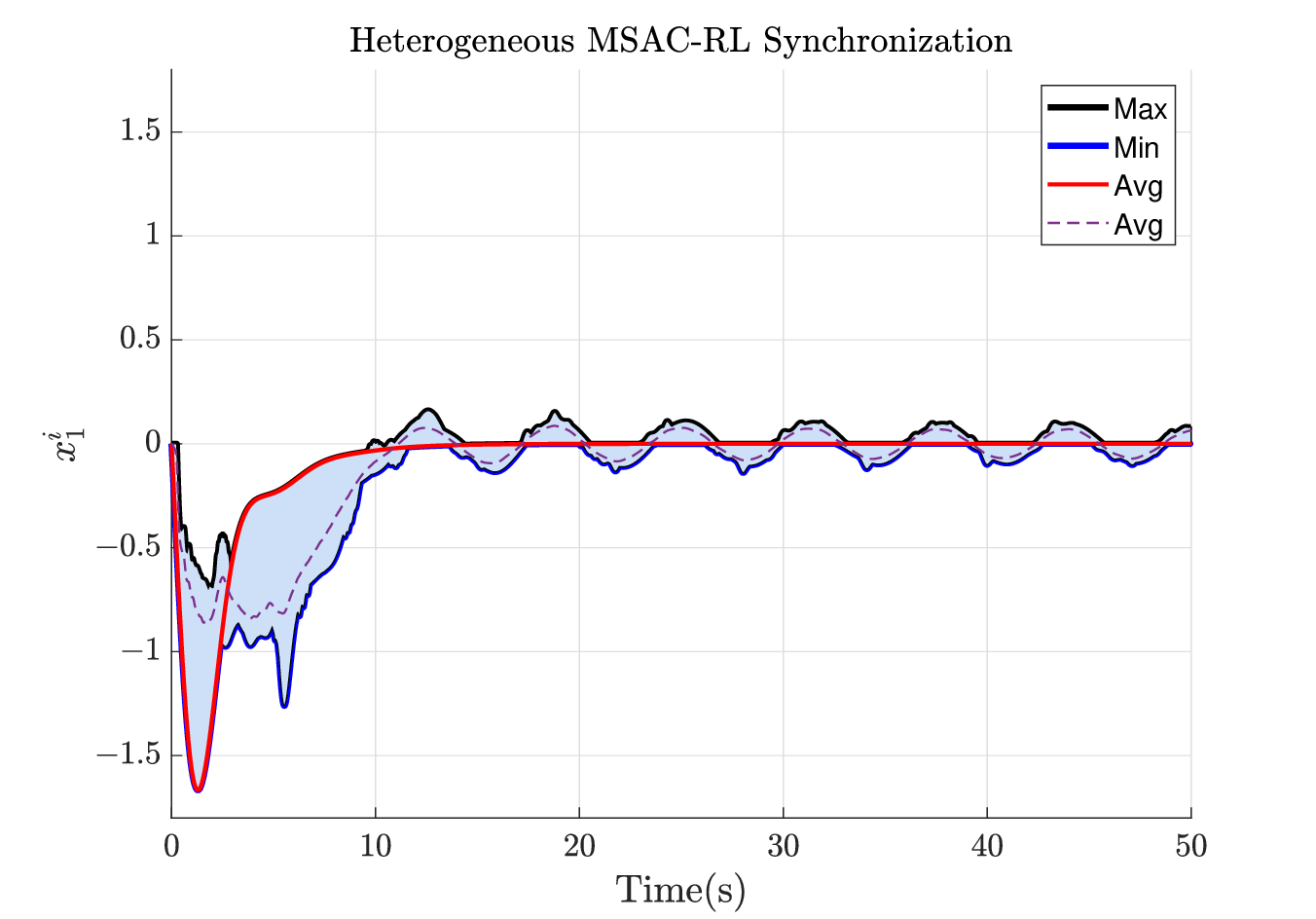}
    \caption{MIMO Multi-agent Synchronization with input magnitude saturation included. The worst agents' response delimits the shaded area, the average response is dotted, and the reference is red.}
    \label{fig:DMSAC_RL}
\end{figure}

Finally, to validate the saturation magnitude algorithm, Figure~\ref{fig:sin_sat} presents the response of an adaptive control without saturation management, including the saturation block in the simulation. The temporal response indicates a divergence in agent states across the network, demonstrating that without the controller, the network cannot be effectively managed.

\begin{figure}[ht]
    \centering
    \includegraphics[width=0.48\textwidth]{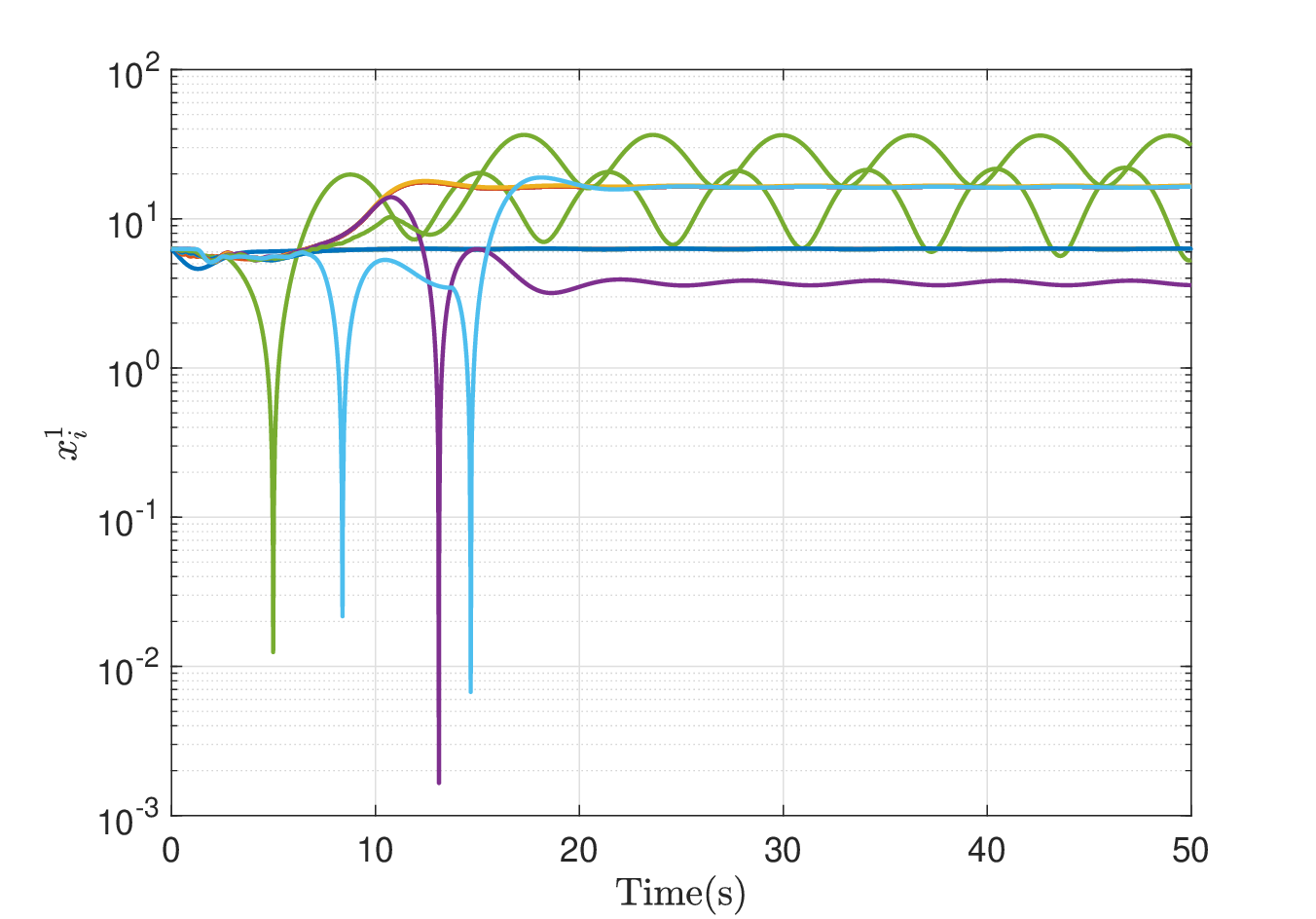}
    \caption{Temporal response of the Adaptive controller without saturation management algorithm included. The response of the states of each of the agents in the network is shown. Legend omitted for space.}
    \label{fig:sin_sat}
\end{figure}

Efficient saturation management is achieved through the controller, adaptive laws \eqref{u_th3}--\eqref{al_msac}, and the magnitude constraint~\eqref{msac_constraint}. Figure~\ref{fig:Magnitude} illustrates the response of the controller error's input magnitude, comparing adaptive techniques without saturation management and with the DMSAC-RL, both with and without saturation management. The blue line represents the error magnitude with DMSAC, while the red line shows the error magnitude without saturation management. In both cases, the saturation parameters were included. The figure demonstrates that without proper saturation management, the controller's response diverges (red) when the saturation block is included, thereby validating the effectiveness of the developed technique.

\begin{figure}[ht]
    \centering
    \includegraphics[width=0.48\textwidth]{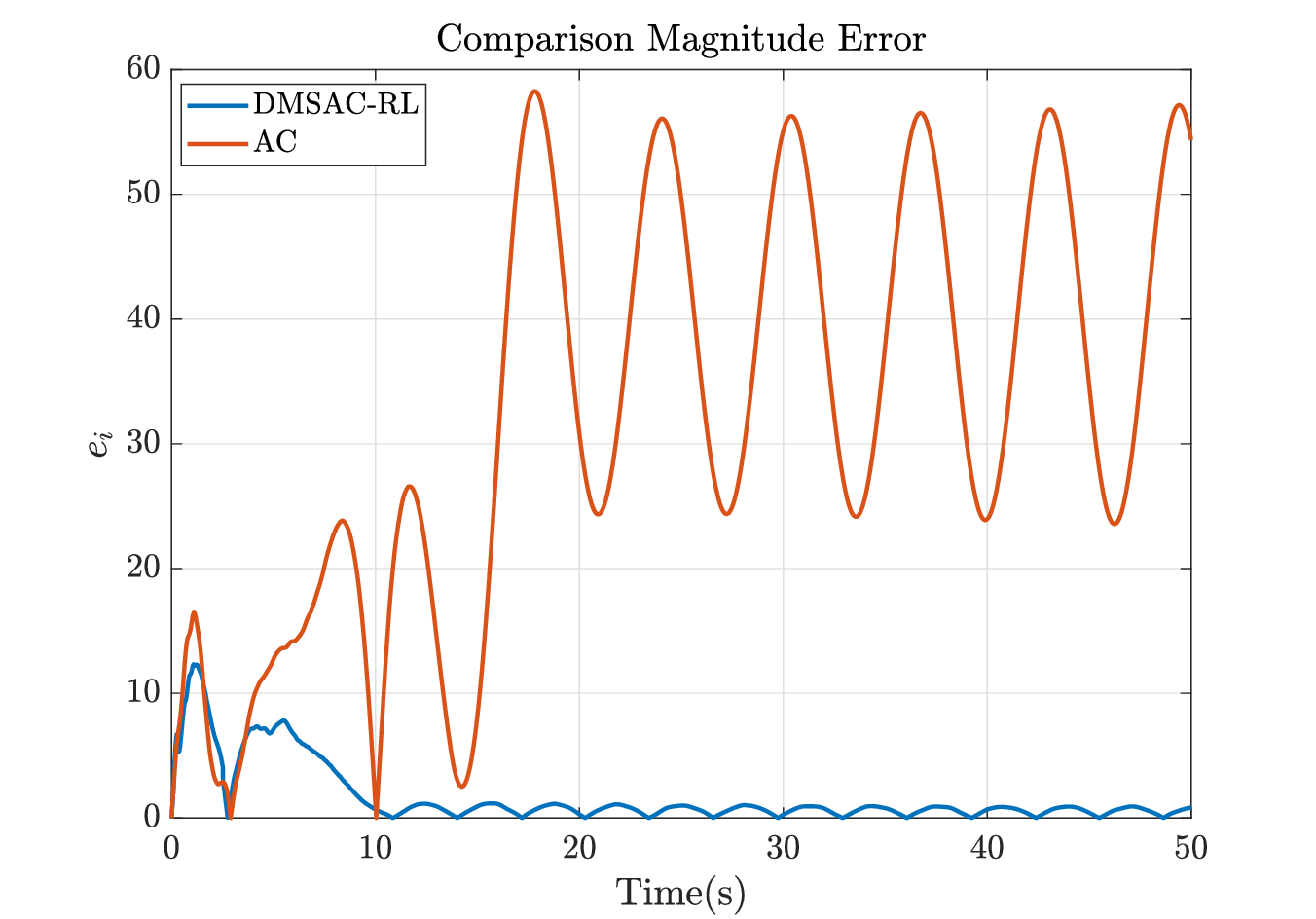}
    \caption{Error input magnitude comparison of heterogeneous synchronization techniques to validate a decrease in the control action without affecting the synchronization. The AC is displayed in red and the DMSAC-RL in blue}
    \label{fig:Magnitude}
\end{figure}

\section{Conclusions} \label{conclusion}
We proposed a distributed MRAC framework for robust and adaptive synchronization of leader-follower networks of heterogeneous nonlinear agents. We assume a pre-trained RL policy is available. This RL policy is trained on a reference model. However, the agents might have different model parameters and input-matched uncertainties. The proposed DMSAC-RL uses an inner loop that directly adjusts the policy for agents and complements an outer loop on augmented input to solve the distributed control problem. A stability analysis has been presented using Lyapunov's theory. We show stability for linear and nonlinear networks with input-matched uncertainties. The stability properties of the system are later extended to the cases of linear systems with input-matched uncertainties and nonlinear networks with no uncertainties. Numerical analysis shows the robustness of the proposed control law for twelve linear pendulums and nonlinear networks with different configurations of input-matched uncertainties in synchronization and tracking scenarios. The proposed method improves the stability properties of the pre-trained RL policy on the studied system. 
Future work will focus on accelerated tracking processes \cite{shi2020learning}, cyclic graphs \cite{Baldi2018cyclic}, time-varying graphs \cite{nedic2015nonasymptotic}, and practical implementations on physical experimental setups.

\bibliography{root}

\end{document}